\pdfoutput=1
\documentclass[final,onecolumn]{svjour3}       

\usepackage{amsmath}
\usepackage{enumitem}
\usepackage{parskip}

\usepackage{setspace}
\usepackage{epsfig}
\usepackage{graphicx}
\usepackage{lmodern}
\usepackage{textcomp}
\usepackage{float}
\usepackage{array}
\usepackage{listings}
\usepackage{subfigure}
\usepackage{multirow}
\usepackage{algorithm,algorithmic}
\usepackage{latexsym,amssymb,amsmath,color}

\newcommand{\dL}[1]{\frac{\partial\mathcal{L}}{\partial{{#1}}}}
\newcommand{\alennote}[1]{{\color{blue} Alen: {#1}}}

\newcommand{\R}{\mathbb{R}}
\newcommand{\X}{\mathcal{X}}
\newcommand{\D}{\mathcal{D}}
\newcommand{\GD}{\Gamma_D}
\newcommand{\GN}{\Gamma_N}

\newcommand{\Np}{{N_p}}

\renewcommand{\vec}[1]{{\mathchoice
                     {\mbox{\boldmath$\displaystyle{#1}$}}
                     {\mbox{\boldmath$\textstyle{#1}$}}
                     {\mbox{\boldmath$\scriptstyle{#1}$}}
                     {\mbox{\boldmath$\scriptscriptstyle{#1}$}}}}
\newcommand{\mat}[1]{\mathbf{{#1}}}

\newtheorem{assumption}{Assumption}

\author{Hayley Guy \and Alen Alexanderian \and Meilin Yu}

\institute{Hayley Guy \at
              Department of Mathematics, North Carolina State University\\
              \email{hguy@ncsu.edu}           \\
           \and
           Alen Alexanderian \at
           Department of Mathematics, North Carolina State University\\
           \email{alexanderian@ncsu.edu}\\
           \and Meilin Yu \at
           Department of Mechanical Engineering, 
           The University of Maryland, Baltimore County\\
           \email{mlyu@umbc.edu}
}

\title{A distributed active subspace method for scalable surrogate modeling
of function valued outputs}

\smartqed  

\newenvironment{myproof}[2]{\paragraph{Proof of {#1} {#2}:}}{\qed}

\begin{document}
\setlength{\abovedisplayskip}{3pt}
\setlength{\belowdisplayskip}{3pt}

\maketitle
\begin{abstract}
We present a distributed active subspace method for training surrogate models
of complex physical processes with high-dimensional inputs and function valued
outputs. Specifically, we represent the model output with a truncated 
Karhunen--Lo\`eve (KL)
expansion, screen the structure of the input space with respect to each KL
mode via the active subspace method, and finally form an overall surrogate
model of the output by combining surrogates of individual output 
KL modes.
To ensure scalable computation of the gradients 
of the output KL modes,
needed in active subspace discovery, we rely on 
adjoint-based gradient computation. 
The proposed method combines benefits of active subspace methods for 
input dimension reduction and KL
expansions used for spectral representation of the output field. 
We provide a mathematical framework for the proposed method and 
conduct an error analysis 
of the mixed KL active subspace approach. Specifically, we provide an 
error estimate that quantifies
errors due to active subspace projection and truncated KL expansion of 
the output.
We demonstrate the numerical performance of the surrogate modeling approach
with an application example from biotransport.

\end{abstract}

\section*{Keywords}
Distributed active subspace; Karhunen--Lo\`eve expansion; Dimension reduction; Function valued outputs; Porous medium flow; Biotransport.

\section{Introduction}\label{sec:intro}

Models with uncertain input parameters are common in modeling of complex
systems.  Computational studies 
such as forward uncertainty propagation,
optimization, or parameter estimation 
require repeated
evaluation of the model.  These tasks become challenging for
expensive-to-evaluate complex models.  To address this challenge, surrogate models
are often used. By approximating the mapping from the uncertain input
parameters to output quantities of interest (QoIs), using a surrogate model,
one can replace expensive model evaluations by inexpensive surrogate model
evaluations.  Examples of surrogate modeling tools include 
polynomial chaos expansion~\cite{GhanemSpanos90,LeMaitreKnio10}, multivariate
adaptive regression splines~\cite{friedman93}, and Gaussian
processes~\cite{RasmussenWilliams06}.

In the present work, we consider surrogate construction for models of the form
\begin{equation}\label{equ:model}
   y = f(\vec{x}, \vec{\xi}), 
\end{equation}
where $\vec{x}$ belongs to a spatial domain and $\vec{\xi} \in \R^\Np$ is a vector of
uncertain parameters.  In our target applications 
$f$ is defined in terms of the solution of a partial differential equation (PDE) that 
is parameterized by a high-dimensional input parameter. Albeit,
the proposed framework can be adapted to more general settings.

A simple approach is 
to
construct a surrogate model pointwise in $\vec{x}$. 
Specifically, discretizing the spatial domain by grid points
$\{\vec{x}_i\}_{i=1}^{N_x}$, we may consider
approximating 
$f(\vec{x}_i, \vec\xi)$, by 
constructing surrogate models $\hat{f}_i(\vec\xi) \approx 
f(\vec{x}_i, \cdot)$, $i = 1, \ldots, N_x$.
This straightforward approach can be useful in some cases, however, the power
of many of the surrogate modeling approaches can be fully realized if one
optimizes them for each $\vec{x}_i$ in the computational grid.
This is a computationally expensive task and might become
prohibitive if the parameter dimension $\Np$ is large (e.g., $\Np \geq 100$),
the model exhibits large variations in its response to 
parametric uncertainties over the spatial domain, or
the grid point number $N_x$ is large, as can happen in two or three-dimensional geometries.

We seek an efficient surrogate modeling approach, for models of the
form~\eqref{equ:model}, whose complexity in terms of the number of model
evaluations does not scale with the dimension $\Np$ of the input parameters.
To enable this, we
need a method that is \emph{structure exploiting}. In particular, we seek to
discover and utilize informative low-dimensional subspaces in the input space
and low-dimensional spectral representations of the model output. The former uses ideas
from active subspace methods~\cite{Constantine15} and the latter uses Karhunen--Lo\`eve (KL) 
expansions~\cite{Loeve77,GhanemSpanos90}. 
This proposed approach decouples the spatial (i.e.
$\vec{x}$) dimensions and those of the random variable $\vec{\xi}$ and exposes
important structures that can be used for building efficient surrogate models. 

\textbf{Related work}.
KL expansions have been used in many works for representing random field
parameters in physics models; see
e.g.,~\cite{Ghanem98,LeMaitreReaganNajmEtAl02,XiuKarniadakis03,LeMaitreKnio04,BabuvskaNobileTempone07,Doostan07,SaadGhanem09,MatthiesKeese05,Graham15KuoKuoNicholsEtAl15,Elman17}.
KL expansions can also be used to represent random field outputs of physics
models, as done in the present work. In models governed by PDEs the output
field often exhibits favorable regularity properties and can be represented
by a KL expansion with a relatively small number of KL terms.  Examples of
this appear for instance in our recent
works~\cite{CleavesAlexanderianGuyEtAl19,ARSY2019} for models governed by
elliptic PDEs.

The active subspace method~\cite{Russi10,constantine2014,Constantine15} has
become a popular approach in recent years for input parameter dimension
reduction and surrogate modeling.  This method seeks to identify important
linear combinations of the input parameters.  
A brief summary of the active subspace method, for approximating 
scalar valued models, has been provided in Section~\ref{sec:background}.
The active subspace method has
been successfully used for uncertainty analysis in models with scalar valued
responses, in a host of engineering applications; examples include scramjet
analysis~\cite{Scramjet15}, wing shape optimization~\cite{ShapeOpt14},
hydrologic modeling~\cite{ConstantineGeo15}, battery
modeling~\cite{ConstantineBattery17}, and kinetic model uncertainty
analysis~\cite{ji2018shared}.  Recently, there have 
been efforts in  
extending the active subspace method to vector valued functions.
In~\cite{zahm2018gradient} the authors find an upper bound for the error of a
ridge function approximation of the original (vector valued) function, and construct an 
approximating function by minimzing that upper bound.
In~\cite{ji2018shared}
an interesting approach is introduced for simultaneously approximating multiple
outputs using a single low-dimensional shared subspace.  The shared subspace is
identified by solving a least-squares system to compute an appropriate
combination of single-output active subspaces.

We also mention another related and popular class of methods for parameter
dimension reduction: variance based~\cite{Sobol:1990,Sobol:2001} and derivative
based~\cite{SobolKucherenko09,KucherenkoIooss17} global sensitivity analysis
(GSA). These methods provide means of identifying unimportant model inputs,
hence reducing the dimension of the input parameter vector.  Derivative based
methods are especially attractive, as they can be used to efficiently screen
for unimportant input parameters, after which a surrogate model can be computed
as a function of a reduced set of parameters.  While GSA approaches have
traditionally been applied to models with scalar outputs, extensions of GSA
methods to vectorial and function valued outputs appear in several recent
works~\cite{GamboaJanonKleinEtAl14,AlexanderianGremaudSmith17,CleavesAlexanderianGuyEtAl19}.

In practice, the active subspace method tends to be very effective in reducing
the input dimension as important \emph{directions} in the input parameter space are
identified. This is, in contrast to seeking 
reduced \emph{parameter subsets}, as identified by GSA
approaches. 
%
%
We note however that utility of GSA goes beyond input parameter dimension
reduction---GSA provides valuable insight regarding a model by identifying key
contributors to model variability. Also, we mention an interesting link between
derivative based GSA and active subspaces through the idea of activity scores
as detailed in~\cite{ConstantineDiaz17}; computing active subspaces provides,
as a byproduct, approximations to derivative based global sensitivity
measures, for scalar valued models.


 
\textbf{Our approach and contributions}.
%
In our proposed approach, we combine active subspaces and KL expansions
to enable efficient surrogate modeling for function valued QoIs.
Specifically, we consider a suitably truncated KL decomposition 
\[
\begin{array}{cc}
f(\vec{x}, \vec{\xi}) \approx  \bar{f}(\vec{x}) +
\sum\limits_{k = 1}^N \sqrt{\lambda_k} f_k(\vec{\xi}) \phi_k(\vec{x}), \\
\end{array}
\]
where $\bar{f}$ is the mean of the process $f(\vec{x}, \vec\xi)$ and
$(\lambda_k, \phi_k)$ are the eigenpairs of the covariance operator of the
process. This way, the model uncertainty is encoded in the KL modes, $f_k$, $k
= 1, \ldots, N$. In many cases a small $N$ can be very effective in
approximating $f$. These KL modes can then be approximated efficiently using
the active subspace approach.  The active subspace approach requires the gradients of
$f_k$, which are scalar functions of the random vector $\vec\xi$. For such
functions, adjoint state approaches provide efficient means of computing the
required gradients. Specifically, gradients can be computed as a cost, in terms
of model evaluations, that does not scale with the dimension of the input
parameter $\vec\xi$. Realizing the proposed approach requires (i) 
a rigorous functional framework upon which efficient 
computational algorithms can be built; and (ii) a systematic 
computational procedure that guides efficient implementations. 
These have been detailed in Section~\ref{sec:method}.


We deploy our proposed framework in the context of flows in biological tissues.
Specifically, in Section~\ref{sec:numerics},
we tackle biotransport in porous tumors with high-dimensional random
inputs (i.e., the permeability field) and random field outputs (i.e.,
the pressure distribution in tumors).
Biological tissues usually have highly heterogeneous, even uncertain, 
material properties~\cite{Frontier:16}. 
The biotransport process such as drug delivery
in tissues thus can exhibit ``unpredictable" behaviors due to the uncertainties
in tissue material properties~\cite{Deb:09CPD}. Since the unpredictability can
adversely affect the effectiveness of therapy, quantifying variability in
biotransport due to uncertain material properties has a high impact on clinical
trial protocol design. However, conducting useful uncertainty quantification
studies on biotransport is challenging due to the high dimension of the
input parameter space, e.g. permeability and
porosity~\cite{AlexanderianZhuSalloumEtAl17}. We find that
efficient-to-evaluate and accurate surrogate models can be computed at a modest
computational cost, with our proposed framework. The presented computational
experiments also indicate that the present method can be used successfully in
more general porous medium flow problems, where Darcy flow is a reasonable
model. Further investigations of the proposed surrogate modeling method in more complex flow
models is subject of our future work.


\renewcommand{\labelitemi}{$\circ$}
The contributions of this article are summarized as follows:
\begin{itemize}
\item We present a mathematical framework and computational method for surrogate modeling
for models with high-dimensional inputs and function valued outputs that
uses 
KL expansions for output dimension reduction, 
active subspaces for input dimension reduction, 
and adjoint based gradient computation
as needed in the active subspace method.

\item We analyze the errors due to active subspace projection and output
KL truncation. The presented analysis provides insight on the interplay 
between these two important sources of errors.

\item We present comprehensive computational results in the context of a
biotransport application problem that test various aspects of the proposed surrogate modeling
method. We mention that the computational studies, to our knowledge,
are the first of a kind in the area of biotransport modeling in tissues with uncertain
heterogeneous material properties.

\end{itemize}

\section{Background on active subspace}\label{sec:background}
One common approach to input dimension reduction is to identify 
a subset of input model parameters 
that are the most important to model variability. This is done 
typically using a local or global sensitivity analysis approach~\cite{KucherenkoIooss17,IoossSaltelli17,PrieurTarantola17}.
The active subspace approach is different; it identifies a set of important
directions in the input parameter space rather than giving importance to one
input over another.  We can rotate the coordinates to align with the directions
of strongest variation. Each direction can be considered as a set of weights
that define a linear combination of all of the inputs. Directions where the
inputs do not vary much are ignored. 

Consider a model 
\[
y = g(\vec{\xi}), \quad \vec{\xi} \in \R^\Np.
\] 
We assume the uncertain parameters $\vec{\xi}$ have an associated probability 
density function $\pi(\vec{\xi})$ that is supported on $\Omega \subseteq \R^\Np$.
Assume $g$ is square integrable and has 
continuous partial derivatives with respect to $\vec{\xi}$.
The following matrix plays a key 
role in active subspace construction:
\begin{equation}\label{equ:matS}
\mat{S} =
\int_\Omega 
\nabla g(\vec\xi)\,\nabla
g(\vec{\xi})^T \pi(\vec\xi) d\vec{\xi},
\end{equation}
where $\nabla g$ is the gradient of $g$.
The matrix $\mat{S}$ is symmetric and positive semi-definite with spectral
decomposition 
\[
\mat{S} = \mat{W}\mat{\Lambda}\mat{W}^T, \quad \mat{\Lambda} = \text{diag}(\lambda_1
,\cdots,\lambda_\Np ). 
\]
The eigenvalues $\lambda_i$'s are sorted in descending order $\lambda_1 \geq
\cdots \geq \lambda_\Np \geq 0$ and $\mat{W}$ contains the orthonormal
eigenvectors of $\mat{S}$. 

The active subspace is determined by the dominant eigenvectors of $\mat{S}$.
Specifically, 
we partition the eigenvalues and eigenvectors 
according to
\begin{equation}\label{equ:part}
\mat{\Lambda}=
\begin{bmatrix} \mat{\Lambda}_1 & \\ & \mat{\Lambda}_2 \end{bmatrix}, \quad \mat{W}=\begin{bmatrix} \mat{W}_1 & \mat{W}_2 \end{bmatrix},
\end{equation}
where 
$\mat{\Lambda}_1 \in \R^{r \times r}$ is a diagonal matrix with the dominant 
eigenvalues on its 
diagonal; and $\mat{W}_1$ contains the corresponding
eigenvectors $\vec{w}_1, \ldots \vec{w}_r$ as its columns. 
These columns span the dominant
eigenspace of $\mat{S}$---\emph{the active subspace}.  

Given
$\vec{\xi}$ we define $\vec{y} = \mat{W}_1^T\vec{\xi} \in \R^r$ and $\vec{z} =
\mat{W}_2^T \vec{\xi} \in \R^{\Np-r}$, and note that
\[
    \vec{\xi} = \mat{W}_1\mat{W}_1^T \vec{\xi} + \mat{W}_2\mat{W}_2^T \vec{\xi}
    = \mat{W}_1 \vec{y} + \mat{W}_2 \vec{z}.
\]

The elements of $\vec{y}$ and $\vec{z}$ are the sets
of active and inactive variables, respectively. 
As discussed in 
detail in~\cite{constantine2014}, the active variables, i.e., elements of 
$\vec{y}$,
are responsible for most of the variations of the function $g$. 

We can consider approximating $g$ in the active subspace, 
$g(\vec{\xi}) \approx 
g(\mat{W}_1\mat{W}_1^T\vec{\xi})$.
It is convenient to define the function $G:\R^r \to \R$ by
\[
G(\vec{y}) = g(\mat{W}_1\vec{y}), \quad \vec{y} \in {\R}^{r},
\] 
and write the approximation to $g$ as 
\[
g(\vec{\xi}) \approx 
G(\mat{W}_1^T \vec{\xi}).
\]

In practice, the function $G$ is typically approximated by a surrogate model.
Specifically, we 
can compute a polynomial regression fit, which we denote by 
$\tilde{G}(\vec{y})$, for the function $G(\vec{y})$. 
Moreover, Monte Carlo sampling is used to approximate 
the matrix $\mat{S}$ in~\eqref{equ:matS}:
\[
\mat{S} \approx \frac{1}{N_s} \sum_{i=1}^{N_s} 
\nabla g(\vec\xi_i) \,\nabla  g(\vec\xi_i)^T. 
\]
Typically, a modest Monte Carlo sample is sufficient for computing
reliable approximations to the dominant eigenvalues and 
eigenvectors of $\mat{S}$.

For readers' convenience, we provide the steps for computing an 
active subspace-based surrogate model
for the function $g$ in Algorithm~\ref{alg:AS}. For further details,
we refer the readers to~\cite{constantine2014}.
\begin{algorithm}
\caption{Computation of an active subspace-based surrogate model for a scalar-valued 
function $g$.}
\label{alg:AS}
\setstretch{1.35}

\begin{algorithmic}
\REQUIRE A set of $N_s$ data points $(\vec\xi_i, g(\vec\xi_i)), i=1,\ldots ,N_s$,
drawn from the law of $\vec\xi$.
\ENSURE Surrogate model $ \tilde{G}(\mat{W}_1^T \vec{\xi})\approx g(\vec\xi) $ 
\STATE Compute the gradients $\vec{D}_i = \nabla g(\xi_i)$, $i = 1, \ldots, N_s$ 
\STATE Compute 
\[
\mat{S} = \frac{1}{N_s} \sum_{i=1}^{N_s} \vec{D}_i \vec{D}_i^T.
\]
\STATE Compute spectral decomposition $\mat{S} =  \mat{W}\mat\Lambda \mat{W}^T$.  
\STATE Based on decay of the eigenvalues, determine the dimension $r$ 
of the active subspace, 
and partition $\mat\Lambda$ and $\vec{W}$ as in~\eqref{equ:part}.

\STATE Compute $\vec{y}_i = \mat{W}_1^T\vec\xi_i$, $i = 1,\ldots, N_s$. 
\STATE Compute a regression fit $\tilde{G}(\vec{y})$ to $G(\vec{y})$ using the 
       data points $(\vec{y}_i, g(\vec\xi_i)), i=1,\ldots,N_s$.
\end{algorithmic}
\end{algorithm}

\section{Active subspace-based surrogate models for function valued QoIs}
\label{sec:method}
In this section, we outline our approach for computing 
an active subspace-based surrogate model for a function valued QoI.
Specifically, we consider models of the form 
\begin{equation}\label{equ:qoi}
y = f(\vec{x}, \vec{\xi}), \quad \vec{x} \in \X, \vec{\xi} \in \Omega,
\end{equation}
where $\X \subset
\R^n$, with $n=2$ or $3$, is a compact set, and
$\Omega \subset \R^\Np$ is a sample space. 
The set $\X$ will be a (sub-) region of a computational domain, 
in our target applications.
Here $\vec{\xi}=(\xi_1, \xi_2, \ldots ,\xi_{\Np})^T$ is
a vector of uncertain parameters, with probability 
density function $\pi(\vec{\xi})$.
We make the following assumptions about $f(\vec{x}, \vec{\xi})$.
\begin{assumption} 
We assume
\renewcommand{\labelenumi}{(\alph{enumi})}
\begin{enumerate}
\item $f \in L^2(\X \times \Omega)$ and $f$ is a mean square continuous process.
\item $\frac{\partial f}{\partial \xi_i}(x,\vec{\xi})$ 
exists for all $x \in \X, \vec{\xi} \in \Omega, i=1,\ldots,\Np$.
\item $\frac{\partial f}{\partial \xi_i}(x,\vec{\xi}) \in 
L^2(\X \times \Omega), i=1,\ldots,\Np$.
\end{enumerate}
\end{assumption}

Letting $\mathbb{E}[\cdot]$ denote expectation with respect 
to $\vec\xi$, the covariance function, $c:\X \times \X \rightarrow \R$, of $f$ 
is defined as
\[
c(\vec{x},\vec{y}) := \mathbb{E}[f(\vec{x},\cdot)f(\vec{y},\cdot)]-
   \mathbb{E}[f(\vec{x},\cdot)]\mathbb{E}[f(\vec{y},\cdot)]
\]
and the associated covariance operator 
$C_f:L^2(\X)\rightarrow L^2(\X)$ is given by 
\[
[C_f u](\vec{x}) := \int_\X c(\vec{x},\vec{y})u(\vec{y})d\vec{y}.
\]
Assumption 1(a) ensures that the covariance function $c$ and the process mean
\[
   \bar{f}(\vec{x}) = \mathbb{E}[f(\vec{x}, \xi)]
\]
are continuous on $\X \times \X$ and $\X$, respectively; 
see e.g.,~\cite[Theorem 7.3.2]{HsingEubank15}.

\subsection{Computational method}
We reduce the dimension of the output by 
computing a low-rank approximation using  
a truncated KL expansion:
\[
\begin{array}{cc}
f(\vec{x}, \vec{\xi}) \approx \hat{f}(\vec{x}, \vec{\xi}) =  \bar{f}(\vec{x}) +
\sum\limits_{k = 1}^N \sqrt{\lambda_k(C_f)} f_k(\vec{\xi}) \phi_k(\vec{x}). 
\end{array}
\]
Here $\lambda_k(C_f)$ and $\phi_k$ are the eigenvalues and eigenvectors of 
the covariance operator $C_f$ 
and $f_k$'s are given by
\[
   f_k(\vec{\xi}) = \dfrac{1}{\sqrt{\lambda_k(C_f)} }\int_\X 
(f(\vec{x}, \vec{\xi})-\bar{f}(\vec{x}))\phi_k(\vec{x}) d\vec{x}.
\]
We refer to the coefficients $f_k$ as the \emph{KL modes} of $f$. 
In the present work, the (generalized) eigenvalue problem, 
\begin{equation}\label{equ:eig_prob}
    C_f \phi_k = \lambda_k(C_f) \phi_k, \quad \int_\X \phi_k(\vec{x})^2\,
                 d\vec{x} = 1, \quad k = 1, 2, \ldots, 
\end{equation}
is solved using
Nystr\"{o}m's method. 
In practice, $N$ can be chosen such that
$(\sum_{k=1}^N \lambda_k(C_f))/(\sum_{k=1}^\infty \lambda_k(C_f)) < 
\mathrm{tol}$, where 
$\mathrm{tol}$ is a user-specified tolerance. 
In many applications of interest, where $f$ is
defined in terms of the solution of a differential equation, 
the eigenvalues $\lambda_k(C_f)$ exhibit
rapid decay, which enables low-rank representations. 
In particular, this is observed in the
application problem considered in the present work.

The next step is to approximate the KL modes
$f_k(\vec{\xi})$ using active subspaces~\cite{constantine2014}.
For each KL mode $f_k$, $k=1,\ldots,N$, we compute an active subspace
by considering  
the symmetric positive semidefinite matrix, $\mat{S}_k \in \R^{\Np \times \Np}$,
defined by 
\begin{equation}\label{equ:Sk}
\mat{S}_k = \int_{\R^\Np} \big(\nabla f_k(\vec\xi)\big)\big(\nabla f_k(\vec\xi)\big)^T \, \pi(\vec{\xi})d\vec{\xi}.
\end{equation}
As before, we compute the spectral decomposition $\mat{S}_k =
\mat{W}_k\mat{\Lambda}_k\mat{W}_k^T$, and partition the eigenpairs 
%
according to
\[\mat{\Lambda}_k=
\begin{bmatrix} \mat{\Lambda}_{k,1} & \\ & \mat{\Lambda}_{k,2} \end{bmatrix},
\quad
\mat{W}_k=\begin{bmatrix} \mat{W}_{k,1} & \mat{W}_{k,2} \end{bmatrix},\]
where $\mat{\Lambda}_{k,1} \in \R^{r_k \times r_k}$ is 
a diagonal matrix with the dominant eigenvalues of $\mat{S}_k$ 
on its diagonal; and $\mat{W}_k$ contains the corresponding eigenvectors.
%
Defining, 
\[
G_k(\vec{y}) = f_k(\mat{W}_{k,1}\vec{y}), \quad \vec{y} \in {\R}^{r_k}, 
k = 1, \ldots, N,
\] 
the KL modes can be approximated by 
\[
f_k(\vec{\xi}) \approx  G_k(\mat{W}_{k,1}^T \vec{\xi}), \quad
k = 1, \ldots, N.
\]
In practice, 
the active subspace-based surrogates for the KL modes $f_k$ 
are constructed by following Algorithm~\ref{alg:AS}, with $g$ replaced by
$f_k$, $k = 1, \ldots, N$. 
Thus, we obtain surrogate models 
\[
    f_k(\vec\xi) \approx \tilde{G}_k(\mat{W}_{k,1}^T\vec\xi), \quad k = 1, \ldots, N.
\]

The overall surrogate model for $f(\vec{x},\vec{\xi})$ is then given by 
\begin{equation}\label{equ:overall}
f(\vec{x}, \vec{\xi}) \approx 
\hat{f}(\vec{x}, \vec{\xi}) =
\bar{f}(\vec{x}) + \sum\limits_{k = 1}^N \sqrt{\lambda_k}
\tilde{G}_k(\mat{W}_{k,1}^T\vec{\xi}) \phi_k(\vec{x}).
\end{equation}
The computational steps for computing $\hat{f}(\vec{x}, \vec\xi)$ can be
divided into three main steps:
\begin{enumerate}
\item Compute the truncated KL expansion of
$f(\vec{x}, \vec\xi)$. 

\item Compute the active subspace-based approximation to KL modes $f_k$, 
$k = 1, \ldots, N$.

\item Form the overall surrogate model as in~\eqref{equ:overall}.
\end{enumerate}

The most computationally challenging part of the above process is the
first step, in which we require an ensemble of model evaluations
$f(\cdot, \vec\xi_i)$, $i = 1, \ldots, N_s$. These model evaluations
will be used to compute the KL expansion of the model $f(\vec{x},\vec\xi)$.
For this, we use Algorithm~1 in~\cite{ARSY2019}, 
that uses Nystr\"{o}m's method to compute $\lambda_k(C_f)$ and
the corresponding eigenvectors $\phi_k(\cdot)$, $k = 1, \ldots, N$.
This process also provides the evaluations of the KL modes, 
   $f_k(\vec\xi_i), k = 1, \ldots, N, \, i = 1, \ldots, N_s$.
As shown in our numerical results, often a modest choice of $N_s$ is
sufficient for obtaining reliable approximations to 
(i) the dominant
eigenpairs, and (ii) KL modes $f_k$.

Notice that for implementing the proposed method, differentiability of
the KL modes is required. Moreover, as seen below, for the 
purposes of error analysis, Lipschitz continuity of the
output KL modes is needed. These requirements can be satisfied
through suitable boundedness assumptions on the partial derivatives of
$f(\vec{x}, \vec\xi)$, as we now explain.
For convenience, and with no loss of generality,
we consider the case where $\bar{f}(\vec{x}) \equiv 0$ and consider 
\begin{equation}\label{equ:bddness}
 F(\vec{\xi}) = \int_\D 
f(\vec{x}, \vec{\xi})v(\vec{x}) d\vec{x},
\end{equation}
where we use a generic element $v \in L^2(\D)$, 
with $\| v \|_{L^2(\D)} = 1$, in place of the eigenvectors $\phi_k$.
This $F$ can be thought of as a generic \emph{unnormalized} KL mode.
In addition to the earlier assumptions on $f$, we also require 
\[
    \left| \frac{ \partial f}{\partial \xi_j} (\vec{x}, \vec{\xi}) \right| \leq b_j(\vec{x}),
    \quad \text{for all } \vec{x} \in \D, \vec{\xi} \in \Omega,
\]
where $b_j$, $j = 1, \ldots, N$, are square integrable.  
\begin{lemma}\label{lem:Lipschitz}
Suppose the process $f(\vec{x}, \vec\xi)$ satisfies Assumption~1 and~\eqref{equ:bddness}.
Then, $F$ is differentiable and is Lipschitz continuous.
\end{lemma}
\begin{proof}
Differentiability of $F$ 
can be shown using the standard arguments of differentiating 
under the integral sign; see e.g.,~\cite{CleavesAlexanderianGuyEtAl19}, 
where it is shown that under the present set of assumptions, 
$\frac{\partial F}{\partial \xi_j} = 
\int_\D \frac{ \partial f}{\partial \xi_j} (\vec{x}, \vec{\xi}) v(\vec{x})\, d\vec{x}$.
We also have 
\begin{equation*}
        |\partial_j F(\vec{\xi})|  \leq \int_\D |\partial_j f(\vec{x}, \vec{\xi})v(\vec{x})|d\vec{x}
\leq \int_\D |b_j(\vec{x})| |v(\vec{x})| \, d\vec{x} \leq \| b_j\|_{L^2(\D)} \| v \|_{L^2(\D)}, 
\end{equation*}
for every $\vec{x} \in \D$.
Using this, it is straightforward to show, 
$| F(\vec{\xi}_1) - F(\vec{\xi}_2) | \leq L \| \vec{\xi}_1 - \vec{\xi}_2 \|$,
for all $\vec{\xi}_1, \vec{\xi}_2 \in \Omega$, 
with $L = \left(\sum_{j=1}^N \| b_j \|_{L^2(\D)}^2\right)^{1/2}$. \qed
\end{proof}

\subsection{Gradient computation}
The present active subspace-based surrogate modeling approach
requires computing the gradient
of the output KL modes. Here, it is more convenient to consider
the ``unnormalized KL modes'', $F_k$, $k = 1, \ldots, N$.
\begin{equation}\label{equ:KL_mode_unnormal}
   F_k(\vec{\xi}) = \int_\X 
(f(\vec{x}, \vec{\xi})-\bar{f}(\vec{x}))\phi_k(\vec{x}) d\vec{x}.
\end{equation}
Finite-difference approximations provide a simple approach, but
will be prohibitive when the input dimension $\Np$ is large and model evaluations are
expensive. 
For field quantities
defined in terms of the solution $u$ of a PDE parameterized by $\vec\xi$, $F_k$ is a functional 
of $u$. This is
an ideal situation for deploying adjoint based gradient computation. 
To illustrate this, we assume the model $f(\vec{x}, \vec\xi)$ is 
a function of the solution $u(\vec{x}, \vec\xi)$ of a PDE. For instance
$f$ can be the restriction of $u$ to a sub-domain, or $f$ can be flux of $u$ through 
a boundary. For illustration, we consider the case where 
the PDE is of the form $\mathcal{A}(\vec\xi) u = b$, where $\mathcal{A}$ denotes 
a differential operator that is parameterized by the uncertain parameter vector $\vec\xi$ 
and $b$ is a source term, and assume 
$f(\vec{x}, \vec\xi) = u(\vec{x},\vec{\xi})$. To simplify the presentation further, 
we consider the discretized problem, where the discretized KL mode is defined by
\[
\vec{F}_k(\vec\xi) = (\vec{u} - \bar{\vec{u}})^T \mat{W} \vec{\phi}_k,
\quad
\text{where}
\quad
\mat{A}(\vec\xi) \vec{u} = \vec{b}.
\]
Here $\vec{u}$ is the discretized state variable, $\mat{W}$ is a diagonal 
matrix with quadrature weights on diagonal, $\vec{\phi}_k$ is the discretized
$k$th eigenvector of output covariance, 
$\mat{A}$ is the discretized
PDE operator, and $\vec{b}$ is the discretized source term.
Computing the gradient of $\vec{F}_k$ with respect to $\vec\xi$ can be 
done using a standard Lagrangian formalism~\cite{Gunzburger03}.
Namely, we define the Lagrangian 
\[
\mathcal{L}(\vec{u}, \vec{\xi}, \vec{q}) = 
(\vec{u} - \bar{\vec{u}})^T \mat{W} \vec{\phi}_k
+ \vec{q}^T \big(\mat{A}(\vec{\xi})\vec{u} - \vec{b}\big),
\] 
where $\vec{q}$ is a Lagrange multiplier. Setting 
$\dL{\vec{q}} = 0$ recovers the state equation, and setting
$\dL{\vec{u}} = 0$ gives the adjoint equation
\[
   \mat{A}^T(\vec\xi)\vec{q} = -\mat{W}\vec{\phi}_k.
\]
Then, the gradient of $\vec{F}_k$ is given by
\[
   \nabla \vec{F}_k(\vec\xi)^T = \dL{\vec\xi} = \vec{q}^T 
   \frac{\partial \mat{A}(\vec\xi)}
   {\partial{\vec{\xi}}}\vec{u}.
\]
Note that evaluation of $\nabla \vec{F}_k$ requires one state (forward)
equation solve and one adjoint equation solve, independently of the dimension
of the uncertain parameter vector $\vec\xi$.  (The forward solves can be reused
across the output KL modes.) Thus, to compute $\nabla \vec{F}_k$, $k = 1,
\ldots, N$, we need one solution of the state equation, and $N$ adjoint solves.
A small $N$, is often sufficient for suitable representations of the output
field due to the rapid decay of the
eigenvalues of the output covariance operator, observed in many applications.

Computing the active subspaces for the output KL modes 
can be done with a set $N_s$ of model evaluations, used across all KL modes.
Thus, the computational cost of active subspace discovery for the output
KL modes is $N_s(1 + N)$, independent of the parameter dimension $\Np$. 
Typically a modest $N_s$ is sufficient, as seen in our numerical results.
Furthermore, the same set of model evaluations can be used for surrogate model
construction for the output KL modes. 
This enables efficient computation of active subspaces for individual KL modes, at 
a cost that does not scale with the dimension of $\vec\xi$.

Notice that the present illustration uses an equation $\mat{A}(\vec\xi)\vec{u} =
\vec{b}$, which is linear in the state variable and nonlinear in the uncertain
parameter $\vec\xi$.  Such an equation can result from discretizing linear (in
state) PDEs that are parameterized by uncertain parameters. The adjoint
approach can more generally be applied to nonlinear PDE models; see
e.g.,~\cite{Gunzburger03}. 

In the present work, we use adjoint based gradient computation for computing the 
gradient of the output KL modes for models governed by elliptic PDEs with a random
coefficient function; see section~\ref{sec:numerics}.

\subsection{Error analysis}
The presented computational strategy involves several approximations
for computing the active subspace-based surrogate for $f(\vec{x}, \vec\xi)$.
In this section, we analyze the errors incurred due to (i) truncation of the 
output KL expansion and (ii) active subspace approximation of
the KL modes. For the purposes of the presented analysis, it is more convenient to work with
unnormalized KL modes~\eqref{equ:KL_mode_unnormal}, and consider 
\begin{equation}\label{equ:KLE_alt}
   f_N(\vec{x}, \vec\xi) = \bar{f}(\vec{x}) + \sum_{k=1}^N F_k(\vec\xi) \phi_k(\vec{x}). 
\end{equation}

The active subspace strategy is flexible regarding the distribution law of 
the random vector $\vec{\xi}$; see e.g.,~\cite{Constantine15}. However, in the present
work, where we consider models with uncertain coefficient functions that 
are modeled using log-Gaussian random fields, 
$\vec\xi$ is a standard Gaussian random vector, i.e., $\vec\xi \sim \mathcal{N}(\vec{0},
\mat{I})$. 
 
As detailed in~\cite{Constantine15}, while Algorithm~\ref{alg:AS}
provides a practical surrogate modeling framework, it is not directly amenable
to theoretical analysis. To enable error analysis, following~\cite{Constantine15}, we define the 
functions $G_k$, used for active subpace projection of the KL modes
in the following way:
\begin{equation}\label{equ:Gk}
     G_k(\vec{y}) = \int f_k(\mat{W}_{k,1} \vec{y} + \mat{W}_{k,2} \vec{z}) \, 
     \pi_{Y | Z}(\vec{z}|\vec{y}) d\vec{z}, 
\end{equation}
where $\pi_{Y | Z}$ is the conditional density, 
\[
\pi_{Y | Z}(\vec{z}|\vec{y}) := 
   \frac{\pi(\mat{W}_{k,1} \vec{y} + \mat{W}_{k,2} \vec{z})}{\pi_Z(\vec{z})}, 
\quad 
\text{with }
\pi_Z(\vec{z}) = \int \pi(\mat{W}_{k,1} \vec{y} + \mat{W}_{k,2} \vec{z})\, d\vec{y}.
\]
In practice, the \emph{marginalized} $G_k$, which for convenience we can denote by
\[
 G_k(\vec{y}) = \mathbb{E}_\vec{z} \{ f_k(\mat{W}_{k,1} \vec{y} + \mat{W}_{k,2} \vec{z})\},
\]
can be approximated by Monte Carlo sampling,
\begin{equation}\label{equ:Gk_hat}
G_k(\vec{y}) \approx \hat{G}_k(\vec{y}) := \frac{1}{N_\text{AS}} \sum_{i=1}^{N_\text{AS}} 
 f_k(\mat{W}_{k,1} \vec{y} + \mat{W}_{k,2} \vec{z}_i).
\end{equation}
However, as seen below, even a very small Monte Carlo Sample (even with
$N_\text{AS} = 1$) can be acceptable. This partly justifies and explains the
effectiveness of Algorithm~\ref{alg:AS}, which can be seen as a special case
of~\eqref{equ:Gk_hat}, with $N_\text{AS} = 1$ and $\vec{z}_1 = 0$.



Recall the approximation of the KL modes $F_k(\vec{\xi}) \approx
G_k(\vec{W}_{k,1}^T\vec{\xi})$.  As a first step in our error analysis, 
we quantify the error in this approximation.
For each $k \in \{1, \ldots, N\}$, we define
\begin{equation}\label{equ:deltak}
\delta_k := \left(\sum_{j = r_k+1}^\Np \lambda_j(\mat{S}_k)\right)^{1/2},
\end{equation}
where 
$\{ \lambda_j(\mat{S}_k)\}_{j=1}^\Np$ are the eigenvalues of $\mat{S}_k$ 
defined in~\eqref{equ:Sk} and $r_k$ is the dimension of the active subspace for the $k$th KL mode $F_k(\vec{\xi})$. 
The following result bounds the error in
approximating the individual KL modes. 
\begin{lemma} \label{lem:basic}
Let $\delta_k$ be as in~\eqref{equ:deltak}, and let $G_k$ and $\hat{G}_k$
be as in~\eqref{equ:Gk} and~\eqref{equ:Gk_hat}. Then,
\begin{enumerate}[label=(\alph*)]
\item $\displaystyle\int_\Omega (F_k(\vec{\xi}) - {G}_k(\vec{W}_{k,1}^T\vec{\xi}))^2 
\pi(\vec{\xi}) d\vec{\xi} \leq \delta_k^2$. 
\item $\displaystyle\int_\Omega (F_k(\vec{\xi}) - \hat{G}_k(\vec{W}_{k,1}^T\vec{\xi}))^2 
\pi(\vec{\xi}) d\vec{\xi} \leq (1+N_\text{AS}^{-1/2})\delta_k^2$. 
\end{enumerate}
\end{lemma}
\begin{proof}
The KL modes $F_k$ are square integrable, mean zero, and by Lemma~\ref{lem:Lipschitz}, they 
are differentiable and Lipschitz continuous. Thus, the first statement follows from 
using~\cite[Theorem 4.3]{Constantine15} 
and the second one follows from~\cite[Theorem 4.4]{Constantine15}. \qed
\end{proof}

%

Next, we consider the error, due to active subspace projection (for individual 
KL modes), in approximating the KL expansion:
\[
e(\vec{x},\vec{\xi}) = | f_N(\vec{x},\vec{\xi}) - \hat{f}_N(\vec{x}, \vec{\xi})|
\]
where $f_N$ is as in~\eqref{equ:KLE_alt}, and
\begin{equation}\label{equ:fhat}
\hat{f}_N(\vec{x},\vec{\xi}):=  
\displaystyle\sum_{k=1}^N {G}_k(\vec{W}_{k,1}^T\vec{\xi}) 
\phi_k(\vec{x})
\end{equation}
its active subspace-based approximation 
where ${G}_k(\vec{W}_{k,1}^T\vec{\xi})$ is our approximation of
the KL modes $F_k(\vec{\xi})$, as defined before. Then let
\[
\bar{e}(\vec{\xi}) = \int_\X e(\vec{x},\vec{\xi}) d\vec{x}.
\]

\begin{theorem}\label{thm:fN_error}
$\mathbb{E}\{\bar{e}\} \leq |\X|^\frac{1}{2} \displaystyle\sum_{k=1}^N \delta_k$,
where $\delta_k$ is as in~\eqref{equ:deltak}.
\end{theorem}
\begin{proof}
See Appendix~\ref{sec:proofs}. \qed
\end{proof}

\begin{remark} \label{rmk:Gk_hat_AS}
Note that in view of Lemma~\ref{lem:basic}(b), if we use $\hat{G}_k$ defined
in~\eqref{equ:Gk_hat}, instead of $G_k$ in~\eqref{equ:fhat}, we can repeat
the argument in proof of Theorem~\ref{thm:fN_error} to get the following estimate:
\[
\mathbb{E}\{\bar{e}\} \leq |\X|^\frac{1}{2} (1+N_\text{AS}^{-1/2})
\displaystyle\sum_{k=1}^N \delta_k.
\] 
\end{remark}

Finally, we consider the overall error of approximating $f(\vec{x}, \vec\xi)$,
due to KL truncation and active subspace projection:
\[
E(\vec{x},\vec{\xi}) = |f(\vec{x},\vec{\xi})-\hat{f}_N(\vec{x},\vec{\xi})|
\]
where $f$ is the original QoI and $\hat{f}_N$ is the active subspace-based approximation as before. 
We consider,
\[
\bar{E}(\vec{\xi}) := \int_\X E(\vec{x},\vec{\xi}) dx.
\]
We have the following result:
\begin{theorem}\label{thm:overall_error}
\begin{equation}\label{equ:estimate}
\mathbb{E}\{\bar{E}(\vec{\xi})\} \leq |\X|^\frac{1}{2}\Bigg[\Bigg(\sum_{k=N+1}^\infty \lambda_k(C_f)\Bigg)^{1/2}+\sum_{k=1}^N \left(\sum_{j = r_k+1}^\Np \lambda_j(\mat{S}_k)\right)^{1/2}\Bigg].
\end{equation}
\end{theorem}
\begin{proof}
See Appendix~\ref{sec:proofs}. \qed
\end{proof}

We note that the first term in~\eqref{equ:estimate} indicates error due to
truncation of the output KL expansion. Recall that $\sum_{k=1}^\infty \lambda_k(C_f) =
\int_\X c_f(\vec{x}, \vec{x}) \, d\vec{x} < \infty$, where the equality is due
to Mercer's Theorem, and the finiteness of the integral is due to continuity of
the covariance operator, which is a consequence of the mean square continuity
assumption on the process.  Therefore, the first term in~\eqref{equ:estimate}
can be made arbitrarily small by taking $N$ sufficiently large.  However,
choosing a large $N$ could entail accumulation of error due to active subspace
projection error in the second term in~\eqref{equ:estimate}; this error,
however, can be controlled by increasing $r_k$.  In practice, in many
applications, a small number of output KL modes (i.e., a small $N$), can be
used to obtain an accurate KL representation for the output. Moreover,
typically low-dimensional (in many cases one- or two-dimensional) active
subspaces can be afforded for approximating the dominant KL modes. We
demonstrate these issues numerically in Section~\ref{sec:numerics}.

\begin{remark}\label{rmk:Gk_hat}
In view of Remark~\ref{rmk:Gk_hat_AS},
if we use $\hat{G}_k$ defined
in~\eqref{equ:Gk_hat}, instead of $G_k$ in~\eqref{equ:fhat}, we can repeat
the argument in proof of Theorem~\ref{thm:overall_error} to get the following estimate:
\[
\mathbb{E}\{\bar{E}(\vec{\xi})\} \leq |\X|^\frac{1}{2}\Bigg[\Bigg(\sum_{k=N+1}^\infty \lambda_k(C_f)\Bigg)^{1/2}+ (1+N_\text{AS}^{-1/2})\sum_{k=1}^N \left(\sum_{j = r_k+1}^\Np \lambda_j(\mat{S}_k)\right)^{1/2}\Bigg].
\]
\end{remark}

\textbf{Additional sources of error}.
The above error analysis only concerns errors assiciated with active subspace
projection and output KL truncation. In practical computations there are a number
of other errors. Most of these errors are related to the active subspace approach used
for approximating the KL modes $F_k$. These include errors in approximating
the eigenvalues and eigenvectors of $\mat{S}_k$'s, incurred due to sample 
average approximation to these matrices and surrogate modeling errors
incurred in approximating $G_k$'s. Errors in approximating $\mat{\Lambda}_k$
and $\mat{W}_k$ are analyzed for instance in~\cite{Constantine15}. Errors
due to surrogate modeling of $G_k$'s are difficult to quantify in general, 
as these errors depend on the choice of surrogate modeling framework. 

There are also further errors related to KL approximation of the 
output field. Namely, we need to approximate the mean field $\mathbb{E}\{f(\vec{x}, \cdot)\}$. 
If simple Monte Carlo sampling is used, approximations of the mean field
exhibit the usual Monte Carlo convergence behavior. However, quasi-Monte Carlo 
approaches can provide efficient means of obtaining more accurate estimates. 
Finally, there will be errors in computing the eigenvalues $\lambda_k(C_f)$ and
the corresponding eigenvectors. These errors are due to (i) sample average
approximation to the output covariance function, and (ii) errors due 
to discretizing the generalized eigenvalue problem~\eqref{equ:eig_prob}.
The discretization errors of course depend on the numerical method
used for solving the eigenvalue problem. For example, if Nystrom's method
is used, as done in the present work, the discretization errors can be 
controlled by the resolution of the computational grid, and the quadrature
method used.

We demonstrate numerically that once a suitable output dimension reduction is
determined, the active subspace approach can be deployed to obtain
approximations to the output KL modes and an overall surrogate model that
captures the statistical properties of $f$ reliably.  We find that this can be
accomplished with an ensemble of function evaluations $\{ f(\cdot,
\vec{\xi}_j)\}_{j = 1}^{N_s}$, with a modest $N_s$.

\section{Application to biotransport in tumors}
\label{sec:numerics}
In this section, we present our computational results in the context of a
biotransport application problem. We begin by describing the governing model in
Section~\ref{sec:model}. Next, we discuss computation of the KL expansion of
the output in Section~\ref{sec:KLE_Bio}. This is followed by our results on
active subspace discovery and surrogate model construction in
Section~\ref{sec:AS_Bio}.  We test the accuracy of the computed surrogate
models in Section~\ref{sec:AS_Bio_Test}. 

\subsection{The governing model}
\label{sec:model}
In this section we describe the biotransport problem we seek to investigate
using the proposed method. We are interested in understanding the impact of
uncertainty in the material properties of cancerous tumors. Specifically we seek
to characterize the uncertainties in the pressure field when a single needle
injection occurs at the center of a spherical tumor with uncertain
heterogeneous structure.
We focus on a 2D cross-section, and consider Darcy's
law constrained by mass conservation in a domain $\D \subset
\R^2$ given by a circle of radius $R_\text{tumor} = 5$ mm, centered
at the origin, with an inner circle
of radius $R_\text{needle} = 0.25$ mm, modeling the injection site, removed.
We denote the inner and outer boundaries of the domain
by $\Gamma_\text{N}$ and $\Gamma_\text{D}$, respectively. 
The following elliptic PDE governs the
fluid pressure $p$:
\begin{equation} \label{equ:2D_Darcy}
\begin{aligned}
-\nabla \cdot \left( \frac{\kappa}{\mu} \nabla p\right) &= 0 \quad \text{in } \D,\\ 
p &= 0 \quad \text{on } \Gamma_\text{D}, \\
\nabla p \cdot \vec{n} &= \frac{Q \mu}{2\pi R_\text{needle} \kappa} 
\quad \text{on } \Gamma_\text{N}.
\end{aligned}
\end{equation}
In this equation, $\kappa$ denotes the absolute permeability field, 
$\mu$ is the fluid dynamic
viscosity, $Q$ is the volume flow rate per unit length, and $\vec{n}$ is
the outward-pointing normal vector.
The nominal values for the above parameters are 
$\kappa = 0.5 \ md$, $\mu = 8.9\times 10^{-4} \ Pa \cdot s$, and $Q = 1 \ mm^2/min$. 
These values are chosen
according to previous investigations of fluid
transport in tumors~\cite{maher:08,ma:12TF,Chen:07}. 
As noted in a number of previous works, 
tumors exhibit complex structures due to their invasive
nature. Generally, tumors consist of loosely organized abnormal cells,
fibers, vasculature, and lymphatics~\cite{Clark:91}, resulting in disordered
tissues with complex heterogeneous structures.   

In the present work, we model the permeability field $\kappa$ by a log-Gaussian
random field as follows.
Let $z(\vec{x}, \omega)$ be a centered Gaussian process; here $\omega \in \Omega$
where $\Omega$ is an appropriate sample space. We assume $z$ has unit pointwise
variance and has correlation function
\begin{equation}
c_z(\vec{x},\vec{y}) = \exp\left\{- \frac{1}{\ell} \| \vec{x}-\vec{y}\|_1\right\}, \quad \vec{x}, \vec{y} \in \D,
\end{equation}
where $\ell > 0$ is the correlation length. In the present study we set the 
correlation length $\ell =1 \ mm$. 
Then, we define the log-permeability field $a = \log\kappa$ according to
\[
a(\vec{x}, \omega) = a_0({\vec{x}}) + \sigma_a z({\vec{x}}, \omega),
\quad {\vec{x}} \in \D, \omega \in \Omega,
\]
where $a_0$ and $\sigma^2_a$ represent the pointwise mean and
variance, respectively. 
We can represent $a(\vec{x},\omega)$ using a truncated KL expansion:
\begin{equation}
a(\vec{x}, \omega) \approx 
\hat{a}(\vec{x}, \omega) := 
a_0(\vec{x}) + \sum_{j=1}^{\Np} \sqrt{\lambda_j(C_a)} \xi_j(\omega) e_j(\vec{x}),
\end{equation}
where $a_0$ is the mean field, $(\lambda_j(C_a), e_j)$ are
the eigenpairs of the covariance operator $C_a$ of $a(\vec{x}, \omega)$, $\Np$ is the input parameter dimension, 
and $\xi_j$ are independent standard normal random variables.                                   
With this parameterization,  
the uncertainty in the (approximate) log permeability field $\hat{a} = \log\kappa$ is completely characterized by 
the random vector $\vec{\xi}=(\xi_1,\xi_2,\ldots, \xi_{\Np})^T$. That is, 
$\hat{a}(\vec{x},\omega)=\hat{a}(\vec{x},\vec{\xi}(\omega))$, and thus, we can consider the (approximate)
log-permeability field as a random process $\hat{a}:\D \times \Omega \to \R$.
The QoI under study here is the pressure field $p(\vec{x}, \vec{\xi})$.

For illustration, two sets of realizations of the permeability field and the
corresponding pressure field with correlation length $\ell=1 \ mm$ are
presented in Figure \ref{fig:log_perm}. We observe large fluctuations in the
permeability field and relatively mild fluctuations in the pressure field. 
\begin{figure}[ht!]\centering
	\includegraphics[width=50mm]{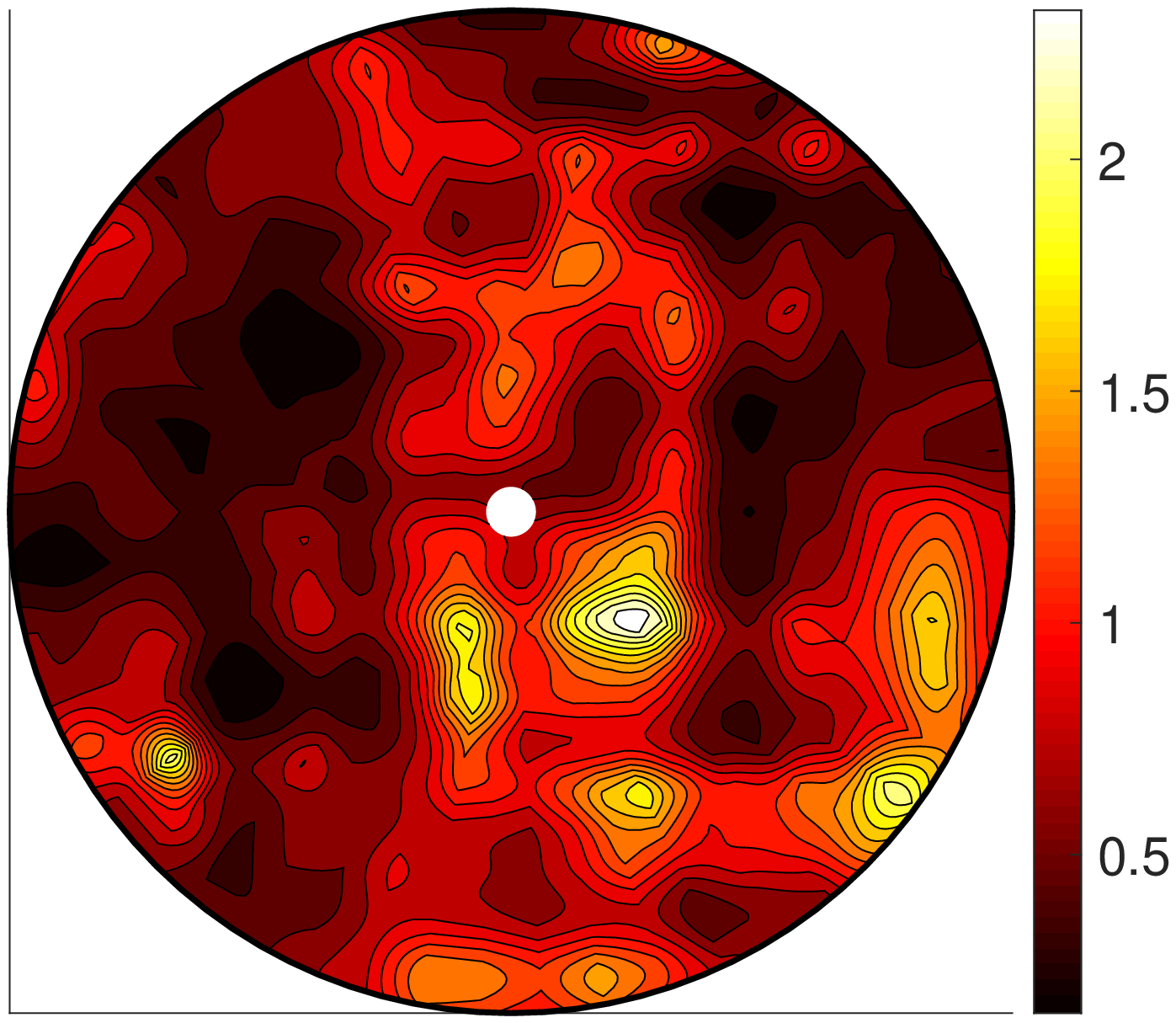}   
	\includegraphics[width=50mm]{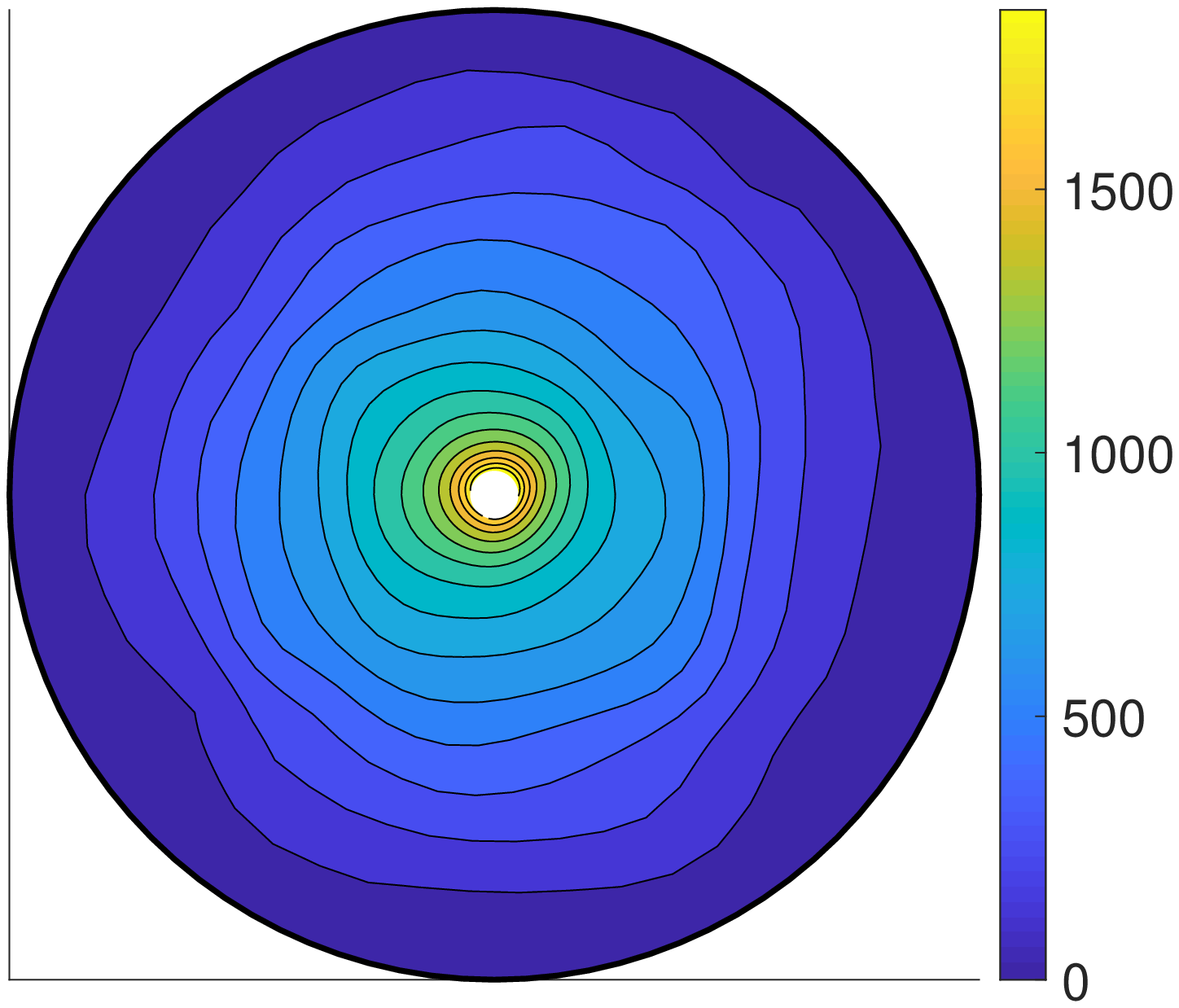}\\ 
	\includegraphics[width=50mm]{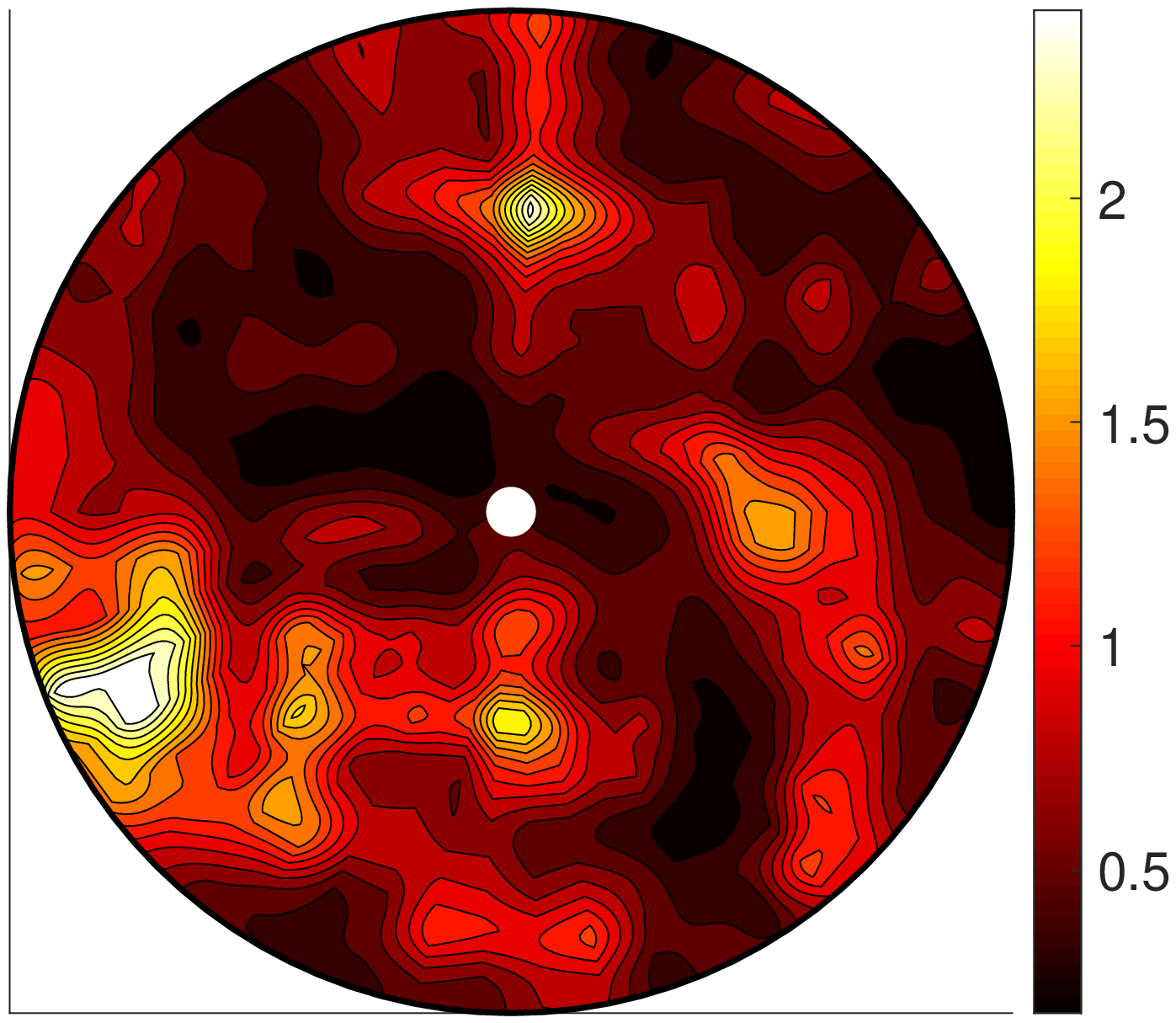} 
	\includegraphics[width=50mm]{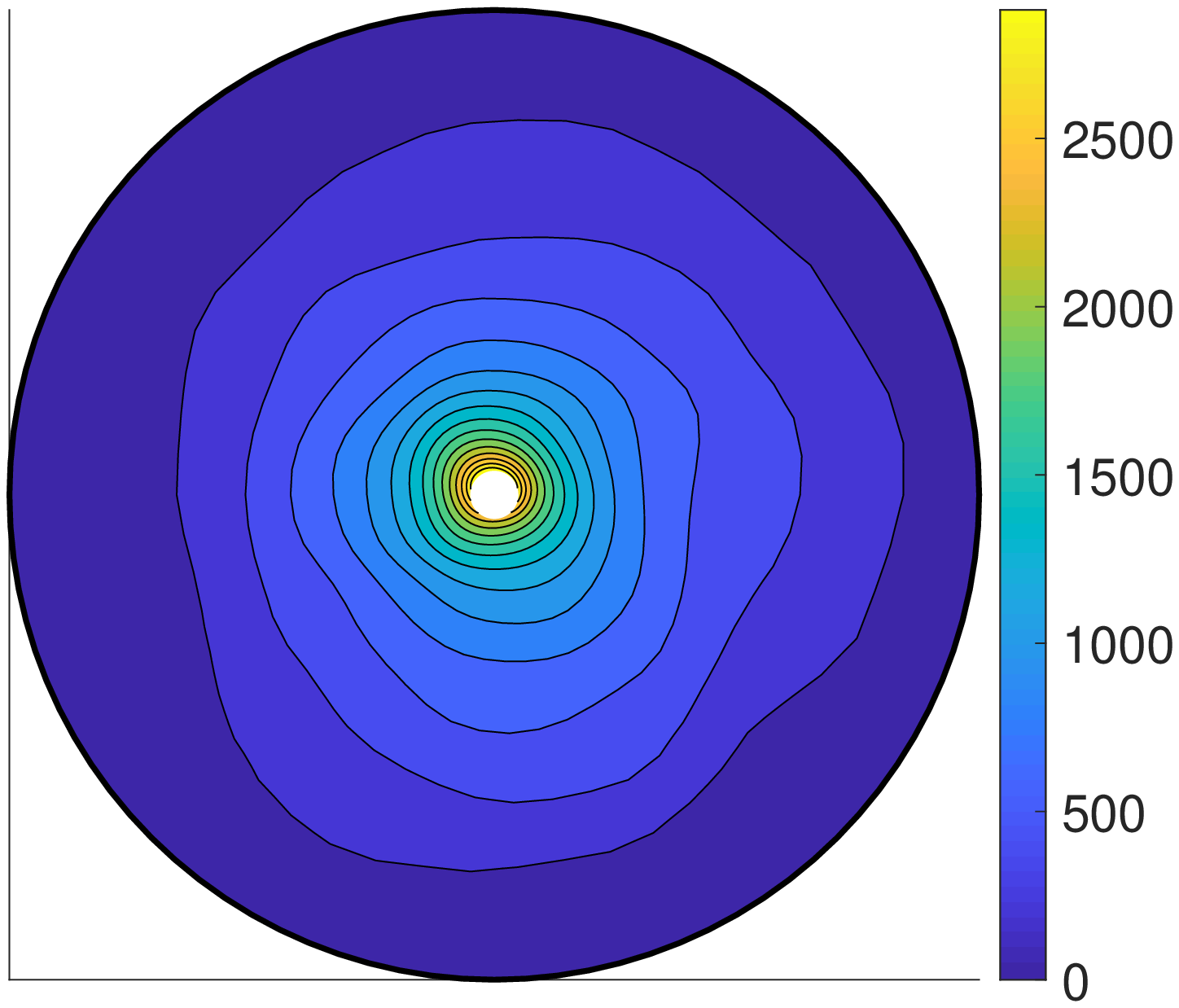}
	\caption{Two realizations of the log-permeability field (left) and the corresponding pressure field (right) using correlation length $\ell=1 \ mm$.}
	\label{fig:log_perm}
\end{figure}

\subsection{Spectral representation of the model output}\label{sec:KLE_Bio}
We represent the QoI, $p$, using a truncated KL expansion
\begin{equation}\label{equ:outputKLE}
p(\vec{x}, \vec{\xi}) \approx   \bar{p}(\vec{x}) +
\sum\limits_{k = 1}^N \sqrt{\lambda_k} p_k(\vec{\xi}) \phi_k(\vec{x}), \qquad 
\text{where}
\qquad
   p_k(\vec{\xi}) = \frac{1}{\sqrt{\lambda_k} }\int_\D 
(p(\vec{x}, \vec{\xi})-\bar{p}(\vec{x}))\phi_k(\vec{x}) d\vec{x}.
\end{equation}
Here $\lambda_k, \phi_k(\vec{x})$ are the eigenvalues and corresponding
eigenfunctions of the covariance operator $C_f$ of $p(\vec{x}, \vec\xi)$.

In Figure \ref{fig:setup} (left) we see that the eigenvalues of the
covariance operator $C_f$ show faster decay than those of the log-permeability
field $a(x,\vec{\xi})$. In Figure \ref{fig:setup} (middle) we compute the
ratio 
$\rho_k = (\sum_{k=1}^N \lambda_k)/\mathrm{Tr}(C_f)$, 
where $\lambda_k$ are the eigenvalues of the covariance operator
$C_f$ and $N$ is the number of KL modes retained. Using just $N = 10$ KL modes
gives us $\rho_k \approx 0.9$ indicating that $90$\% of the variance is captured. If
we add five more modes so that $N=15$ then we capture nearly 95\% of the
average variance in the model output $p$. Figure \ref{fig:setup} (right)
shows the first $40$ eigenvalues of the covariance operator $C_f$, for a number of different sample sizes $N_s$ used to approximate the covariance function of $p$. Note that 
using $N_s=300$ samples we can 
approximate the dominant eigenvalues reasonably well. In
the current study we set $\Np$ the input parameter dimension to be $200$ and to
retain $N=15$ KL modes. 
From the results that follow we will see that
we can achieve significant dimension reduction for the input parameter.

\begin{figure}[h!]
\centering
\includegraphics[width=.3\textwidth]{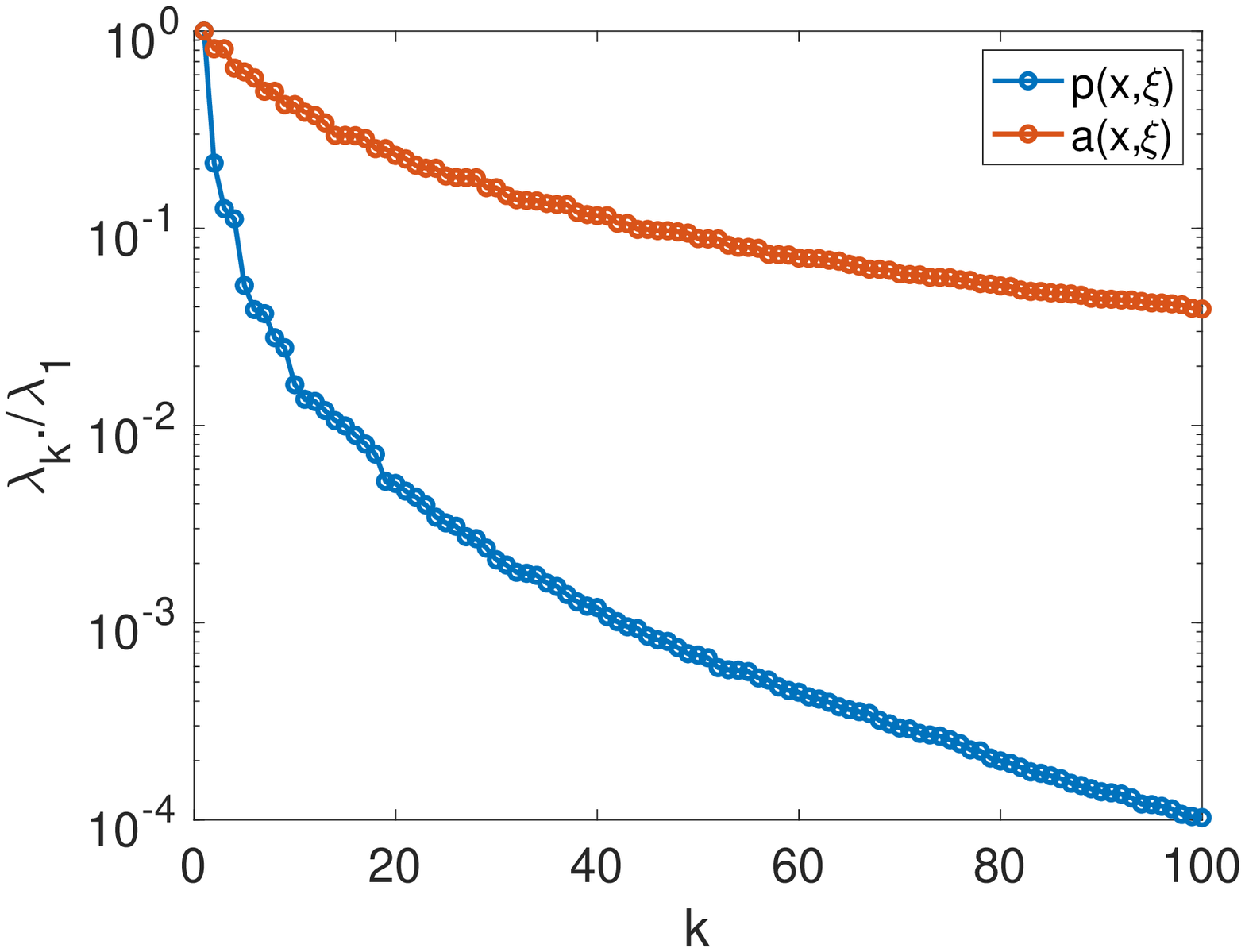}    
\includegraphics[width=.275\textwidth]{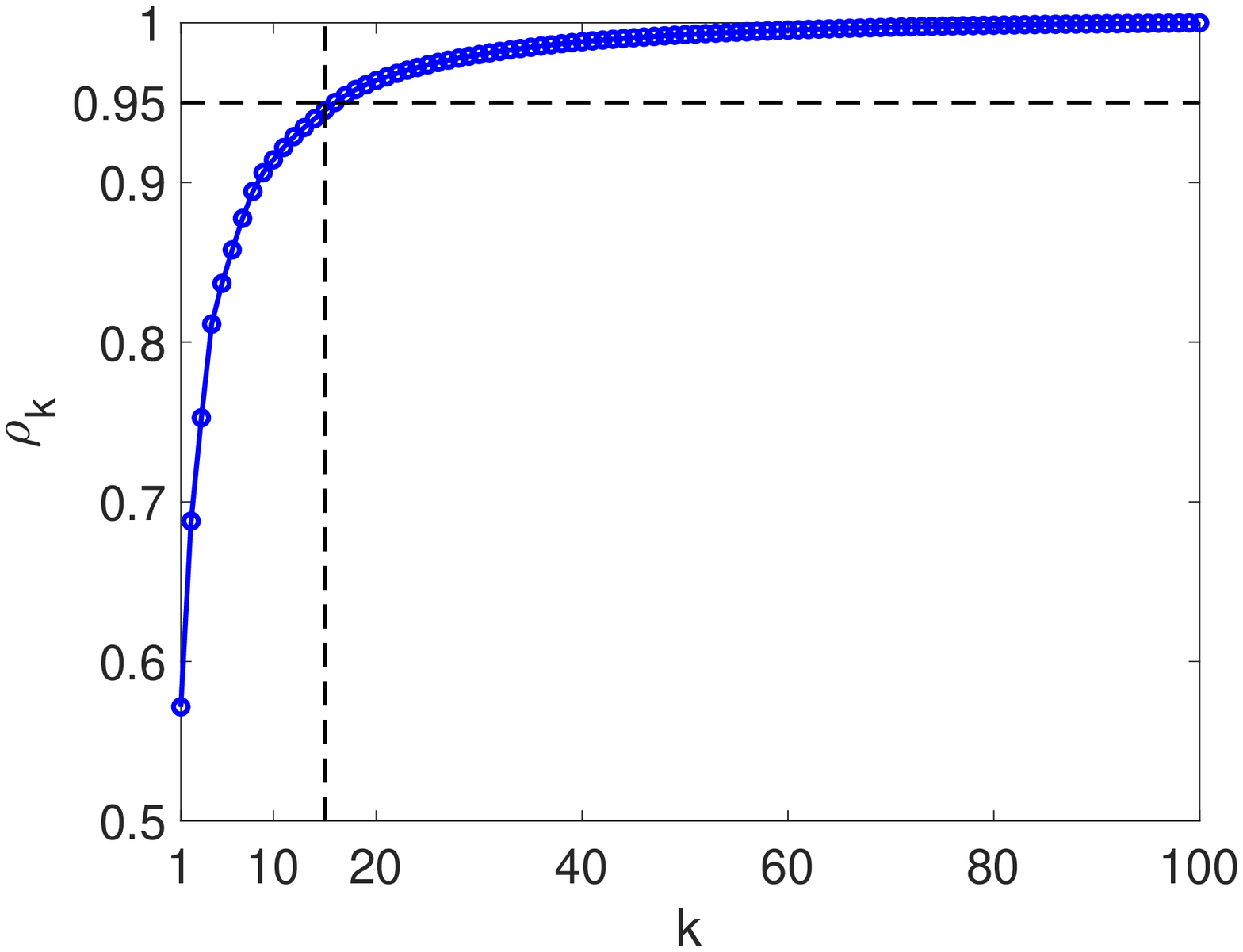} 
\includegraphics[width=.3\textwidth]{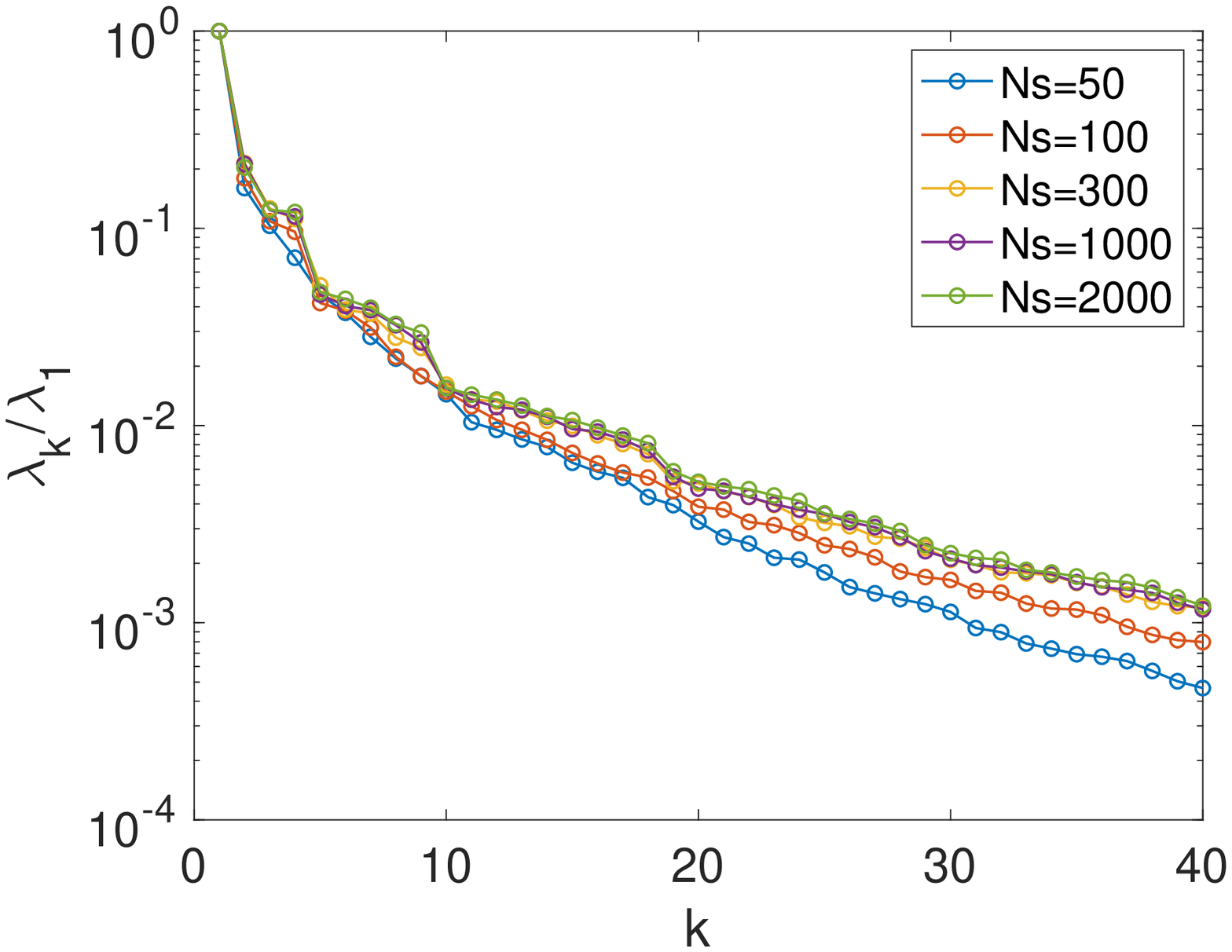}
\caption{Eigenvalue spectrum of $a(\vec{x},\vec{\xi})$ versus 
$p(\vec{x},\vec{\xi})$ (left).
Ratio showing saturation of average variance for $p(\vec{x},\vec{\xi})$ (middle). First $40$ eigenvalues of $C_f$ for different sample sizes $N_s$ (right).}
\label{fig:setup}
\end{figure}

\newcommand{\fac}{\frac{1}{\sqrt{\lambda_k(C_f)}}}
\newcommand{\ut}[1]{\tilde{{#1}}}

\subsection{Active subspace discovery and surrogate model construction}\label{sec:AS_Bio}
To construct the active subspace-based surrogate model for $p(\vec{x}, \vec\xi)$, 
we need the gradients of the output KL modes
defined in~\eqref{equ:outputKLE}. For this, we use the adjoint method.
The adjoint based
expression for $p_k(\vec{\xi})$ can be derived  
using a formal Lagrange approach; see, e.g.,~\cite{Gunzburger03}.
For a basic derivation of the gradient of the output KL modes for 
models governed by elliptic PDEs, we refer the reader 
to~\cite{CleavesAlexanderianGuyEtAl19}, where output KL expansions
are used within the context of derivative based global sensitivity analysis.
The adjoint based expression for the partial derivatives of
$p_k$'s are given by,
\begin{equation}\label{equ:grad}
   \frac{\partial p_k}{\partial \xi_j} = 
   \sqrt{\lambda_j(C_a)} \int_\D 
   e_j(\vec{x}) e^{\hat{a}(\vec{x}, \vec{\xi})}  \nabla p(\vec{x}) 
   \cdot \nabla q(\vec{x}) \, d\vec{x}.
\end{equation}
where $p$ is the solution of the (forward) PDE~\eqref{equ:2D_Darcy} and
$q$ is the solution of the adjoint equation:
\begin{equation}\label{equ:adj}
    \begin{aligned}
      -\nabla \cdot \big(\frac\kappa\mu  \nabla q\big) &= -\fac \phi_k \quad \text{ in }\D, \\
                    q  &= 0 \quad \text{ on } \GD, \\
                    \nabla q \cdot n &= 0 \quad \text{ on } \GN.
    \end{aligned}
\end{equation}
Note that in the forward and adjoint equation, we 
let the permeability field be $\kappa(\vec{x}, \vec\xi) = e^{\hat{a}(\vec{x},\vec\xi)}$.

Guided by the results in the previous subsection, 
we focus on the first $N = 15$ output KL modes.
For each $k \in \{1, \ldots, N\}$, we generate a sample
$\{ \nabla p_k(\vec\xi_i)\}_{i=1}^{N_s}$, with $N_s = 300$. 
For $\{ \vec\xi_i\}_{i=1}^{N_s}$, we used the same set of $300$ 
parameter samples used in computing the KL expansion of the output.
Using~\eqref{equ:grad},
\[
\nabla p_k(\vec\xi_i) =
\begin{bmatrix}
\displaystyle
\sqrt{\lambda_1(C_a)} \int_\D
   e_1(\vec{x}) e^{\hat{a}(\vec{x}, \vec{\xi}_i)}  \nabla p(\vec{x}, \vec\xi_i)
   \cdot \nabla q(\vec{x}, \vec\xi_i) \, d\vec{x}
\\
\displaystyle
\sqrt{\lambda_2(C_a)} \int_\D
   e_2(\vec{x}) e^{\hat{a}(\vec{x}, \vec{\xi}_i)}  \nabla p(\vec{x}, \vec\xi_i)
   \cdot \nabla q(\vec{x}, \vec\xi_i) \, d\vec{x}
\\
\vdots
\\
\displaystyle
\sqrt{\lambda_{N_p}(C_a)} \int_\D
e_{N_p}(\vec{x}) e^{\hat{a}(\vec{x}, \vec{\xi}_i)} \nabla p(\vec{x}, \vec\xi_i)
\cdot \nabla q(\vec{x}, \vec\xi_i) \, d\vec{x}
\end{bmatrix}, \quad i = 1, \ldots, N_s.
\]
To compute these, we reuse the model evaluations $\{p(\vec{x},
\vec\xi_i)\}_{i=1}^{N_s}$, from the computation of the output KL expansion earlier; the
adjoint variables $q(\cdot, \vec\xi_i)$ are computed by solving the adjoint
equation~\eqref{equ:adj}, with $\kappa = e^{\hat{a}(\vec{x},\vec\xi_i)}$, $i =
1, \ldots, N_s$. 

Using the gradient samples, we 
approximate the matrix $\mat{S}_k$ defined in \eqref{equ:Sk} for each KL
mode, $k=1,\ldots,N$:
\[
\mat{\hat{S}}_k = \frac{1}{N_s} \sum_{i=1}^{N_s} \nabla p_k(\vec\xi_i) 
\nabla p_k(\vec\xi_i)^T.
\]
In each case, we consider 
the corresponding spectral decomposition $\mat{\hat{S}}_k =
\mat{\hat{W}}_k\mat{\hat{\Lambda}}_k\mat{\hat{W}}_k^T$. 
To identify the active
subspace and where to partition the eigenpairs we examine the 
spectrum of $\mat{\hat{S}}_k$. 
As an illustration, in Figure \ref{fig:eigs}~(top) we present the
spectrum for the first three output KL modes. Visually we observe a gap between the
first and second eigenvalue for each of the modes indicating 
one-dimensional active subspaces. In Figure~\ref{fig:eigs}~(bottom), we show the corresponding
sufficient summary plots (SSPs) for the corresponding modes. 
A sufficient summary plot here is a scatter plot of the active variables $y=\mat{W}_{k,1}^T\vec\xi$ 
versus the output KL modes $p_k(\vec\xi)$. We observe a strong univariate trend which
further indicates the presence of one-dimensional active subspaces. To provide a
consistent truncation approach, we can use a threshold $\gamma$ and choose the
dimension $r_k$ of the active subspace according to  
$\lambda_{k,1}/\lambda_{k,r_{k+1}} > \gamma$.  
In the current study we use $\gamma = 10$.  Using this approach, we identified a
one-dimensional active subspace for each of the first $15$ output modes, expect
modes $11$ and $12$ where two-dimensional active subspaces were identified.  To
illustrate, we report the spectrum of $\mat{S}_{12}$ and the SSP for $p_{12}$
in Figure~\ref{fig:mode12}.  Based on the determined values of $r_k$, we
partition  
\[
\mat{\Lambda}_k=
\begin{bmatrix} \mat{\Lambda}_{k,1} & \\ & \mat{\Lambda}_{k,2} \end{bmatrix},
\quad
\mat{W}_k=\begin{bmatrix} \mat{W}_{k,1} & \mat{W}_{k,2} \end{bmatrix}, 
\quad k = 1, \ldots, N.
\]
$\mat{\Lambda}_{k,1}$ contains the dominant 
eigenvalues and $\mat{W}_{k,1}$ the corresponding eigenvectors.

\begin{figure}[h!]\centering
		\begin{tabular}{ccc}
			\includegraphics[width=.3\textwidth]{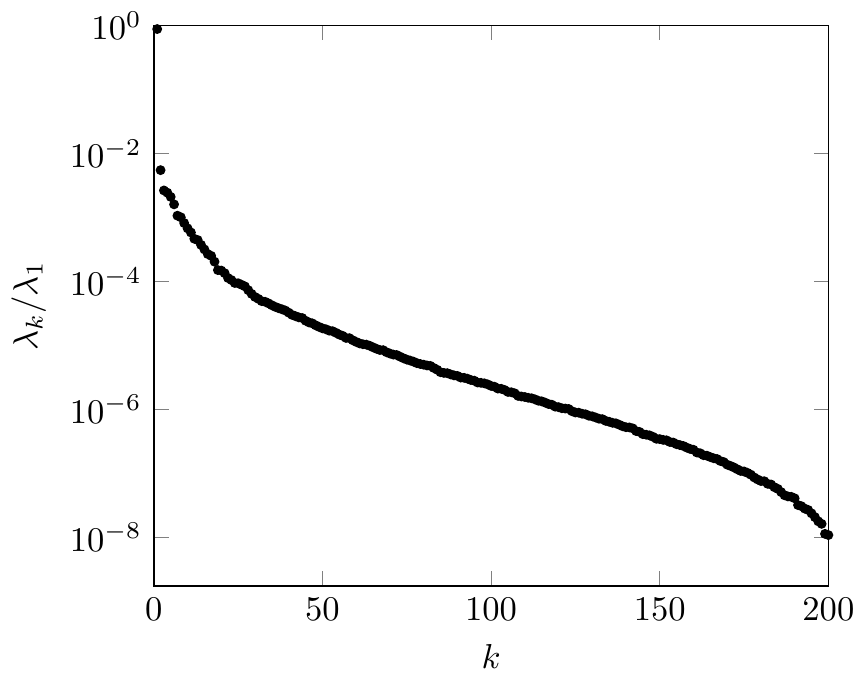} &   
			\includegraphics[width=.3\textwidth]{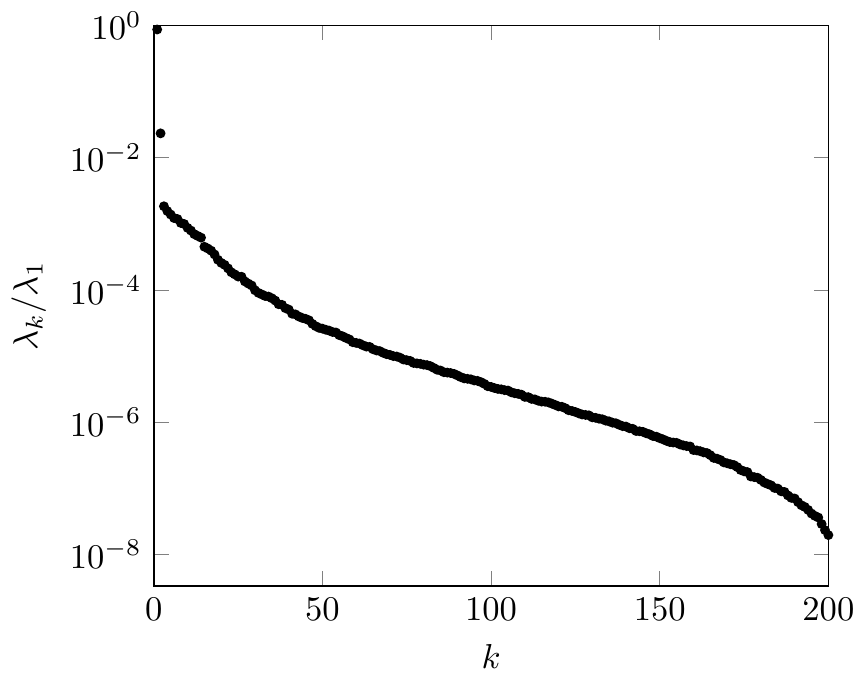} &
			\includegraphics[width=.3\textwidth]{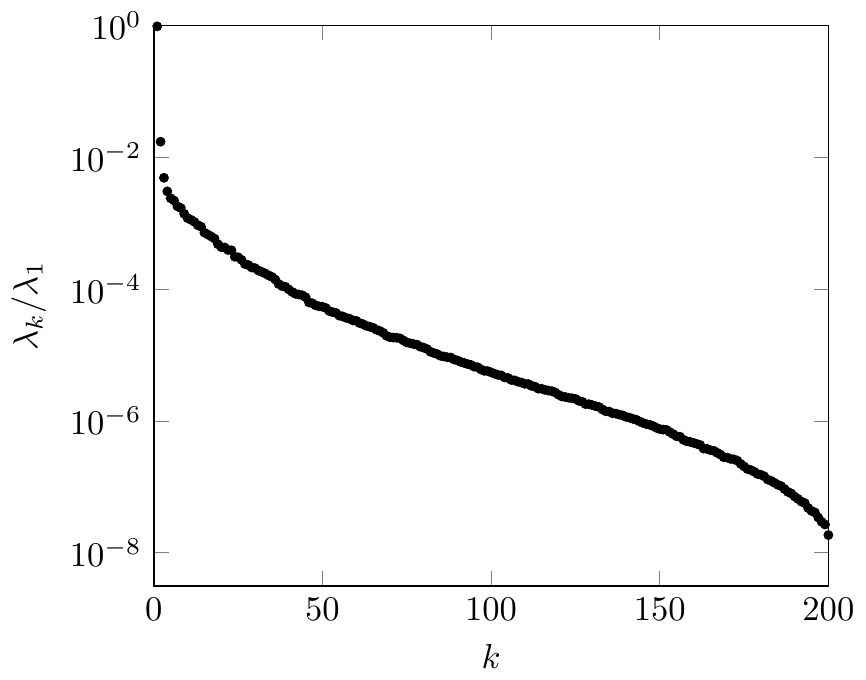} \\
			\includegraphics[width=.3\textwidth]{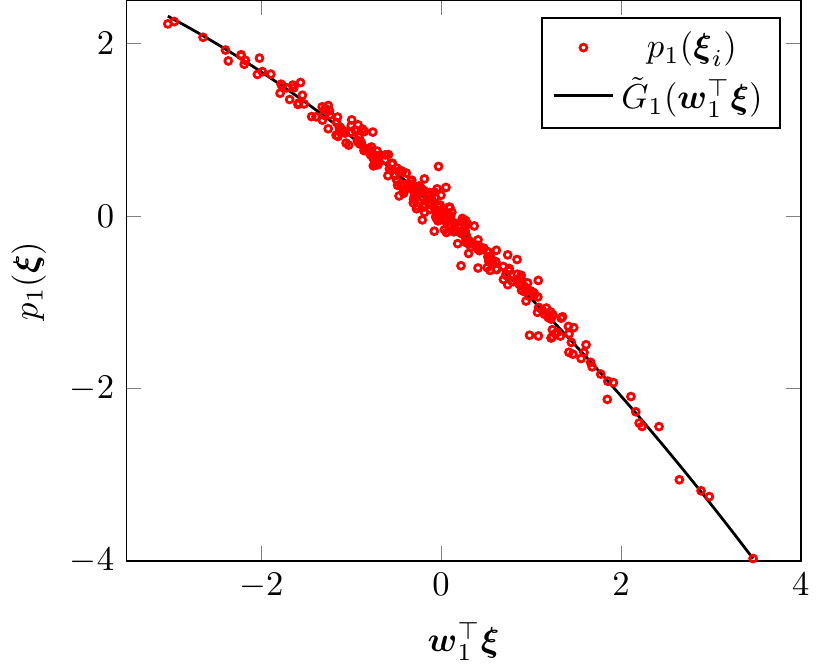} &   
			\includegraphics[width=.3\textwidth]{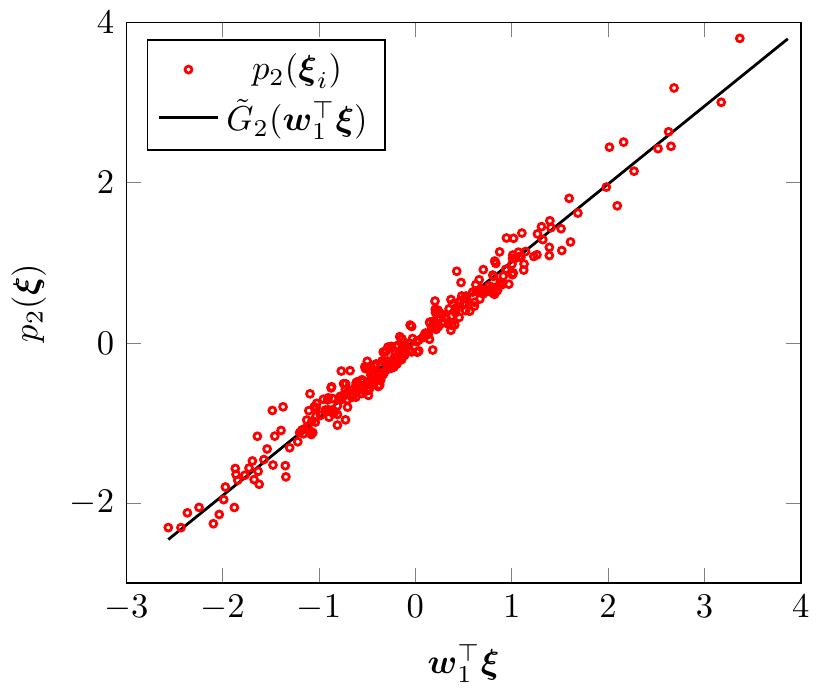} &
			\includegraphics[width=.3\textwidth]{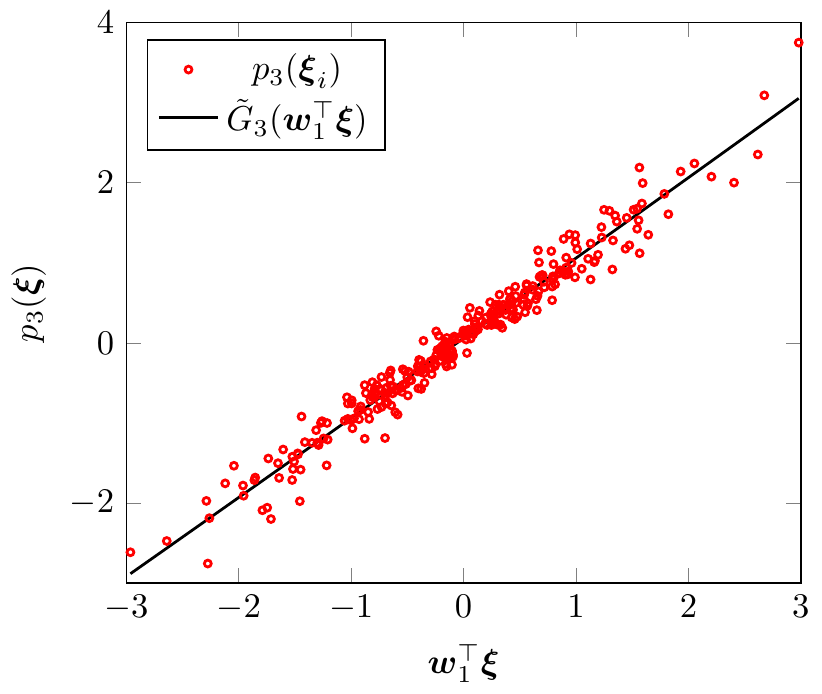} 
		\end{tabular}
	\caption{
Top: eigenvalues of the matrix $\mat{\hat{S}}_k$ for $k=1$ (left) 
$k=2$ (middle), and $k=3$ (right);  
bottom: SSPs for output KL modes $p_1$ (left), 
$p_2$ (middle), and $p_3$ (right).}
\label{fig:eigs}
\end{figure}
\begin{figure}[h!]\centering
\includegraphics[width=.35\textwidth]{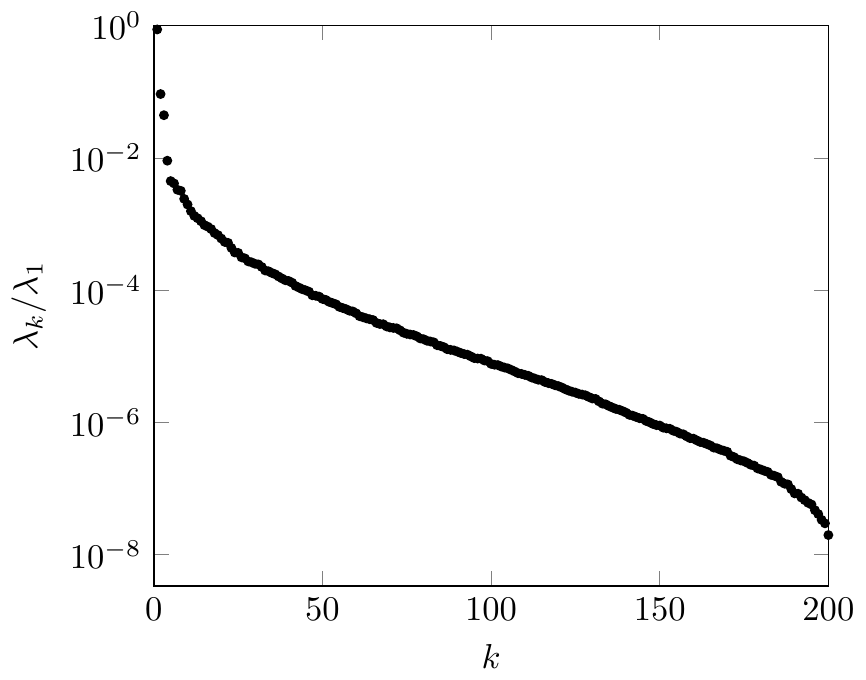}
\includegraphics[width=.45\textwidth]{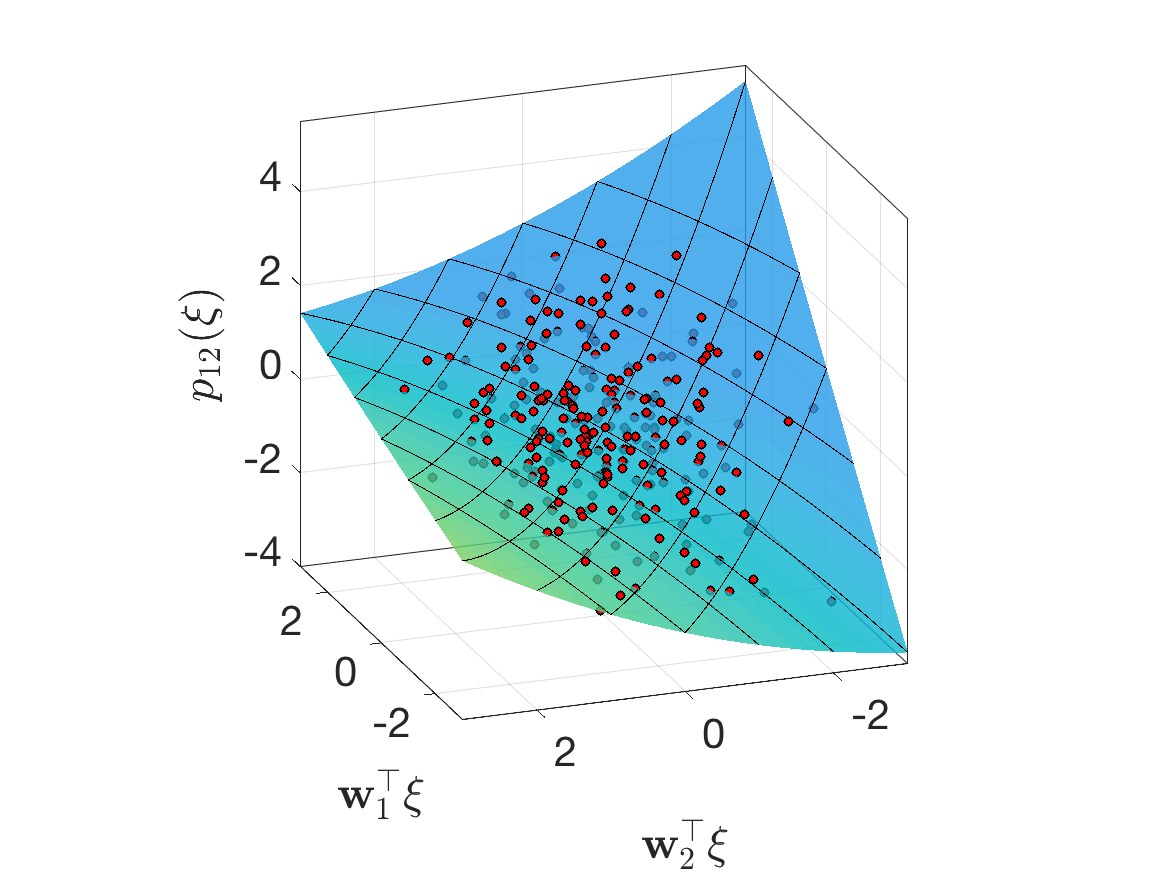}
\caption{The normalized eigenvalues of output mode $p_{12}$ (left)
and the 2D sufficient summary plot corresponding to $p_{12}$ (right). Note that the first two eigenvalues are very close to each other.}
\label{fig:mode12}
\end{figure}


Recall that the active subspace approach essentially seeks ``important
linear combinations'' of the input parameters. To illustrate this,
in 
Figure~\ref{fig:evecs}, we present the components of the
dominant eigenvector for the first three output KL modes $p_1,p_2$ and $p_3$.
The smaller inset plot shows the first 50 components of the dominant
eigenvector for $p_1,p_2$ and $p_3$.  The magnitude of the components give us a
sensitivity measure for each of the input parameters. A large component value
indicates that that particular input is important in defining the direction of
most variation in our function $p$. We note that  
the first output KL mode is sensitive to a few components of the input parameter vector.
In contrast, the second and especially the third
output KL modes show sensitivity to a larger number of components of $\vec\xi$.

Finally, we note that the present results indicate a significant dimension
reduction. The dominant output KL modes, each a function of 200 parameters, 
can be approximated in one or two dimensional active subspaces.

\begin{figure}[h!]\centering
                        \includegraphics[width=.32\textwidth]{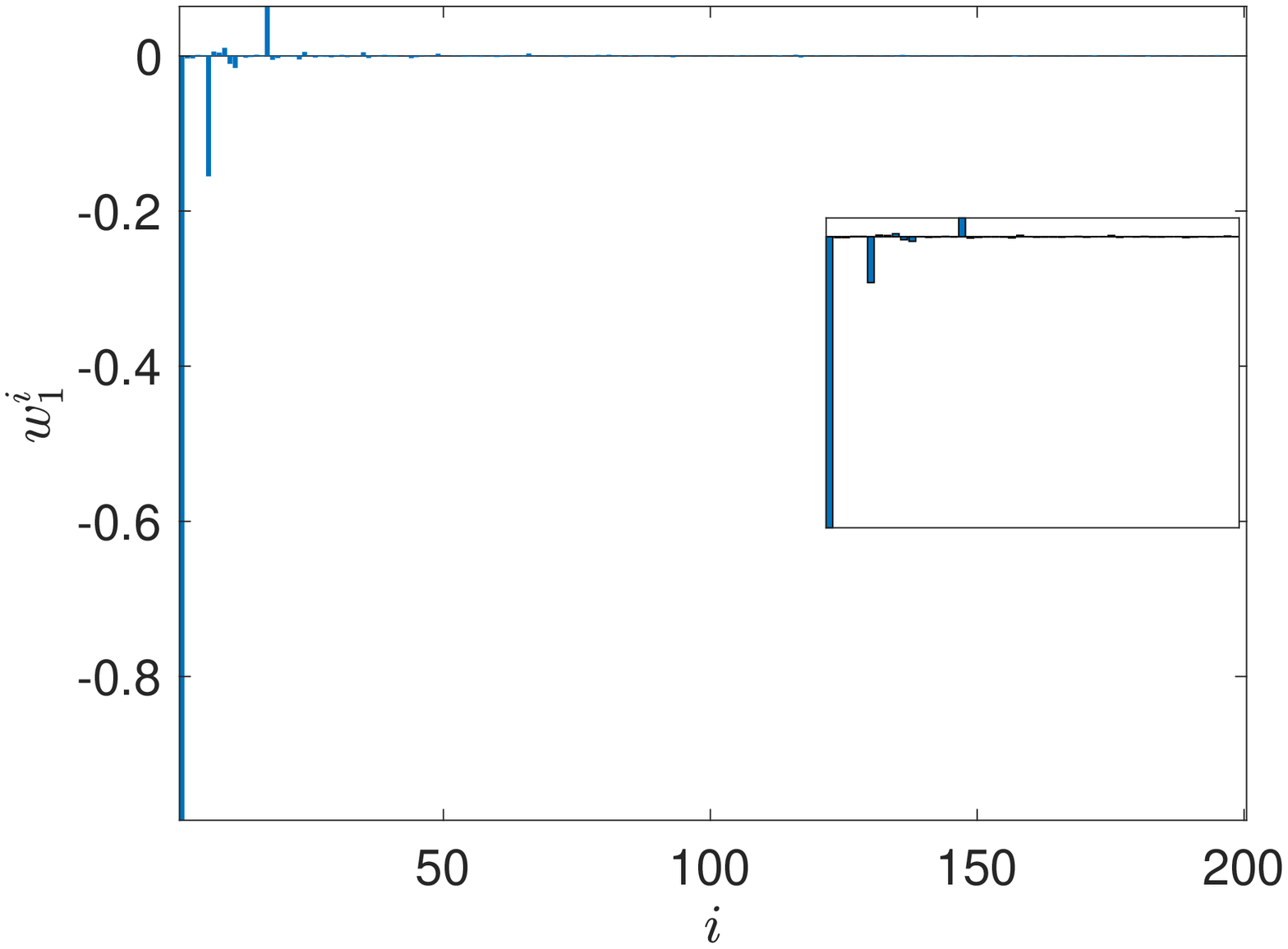} 
                        \includegraphics[width=.32\textwidth]{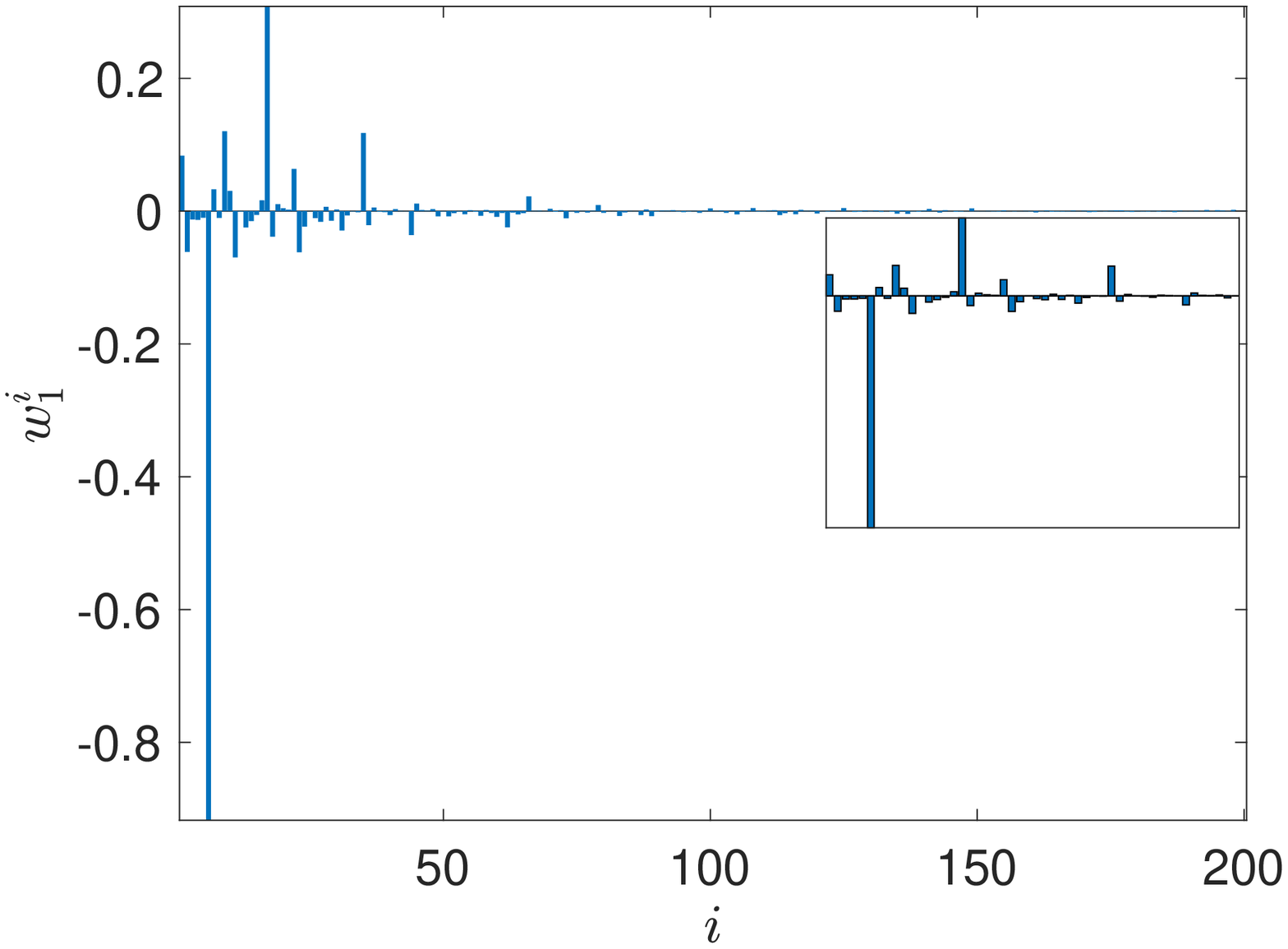} 
                        \includegraphics[width=.32\textwidth]{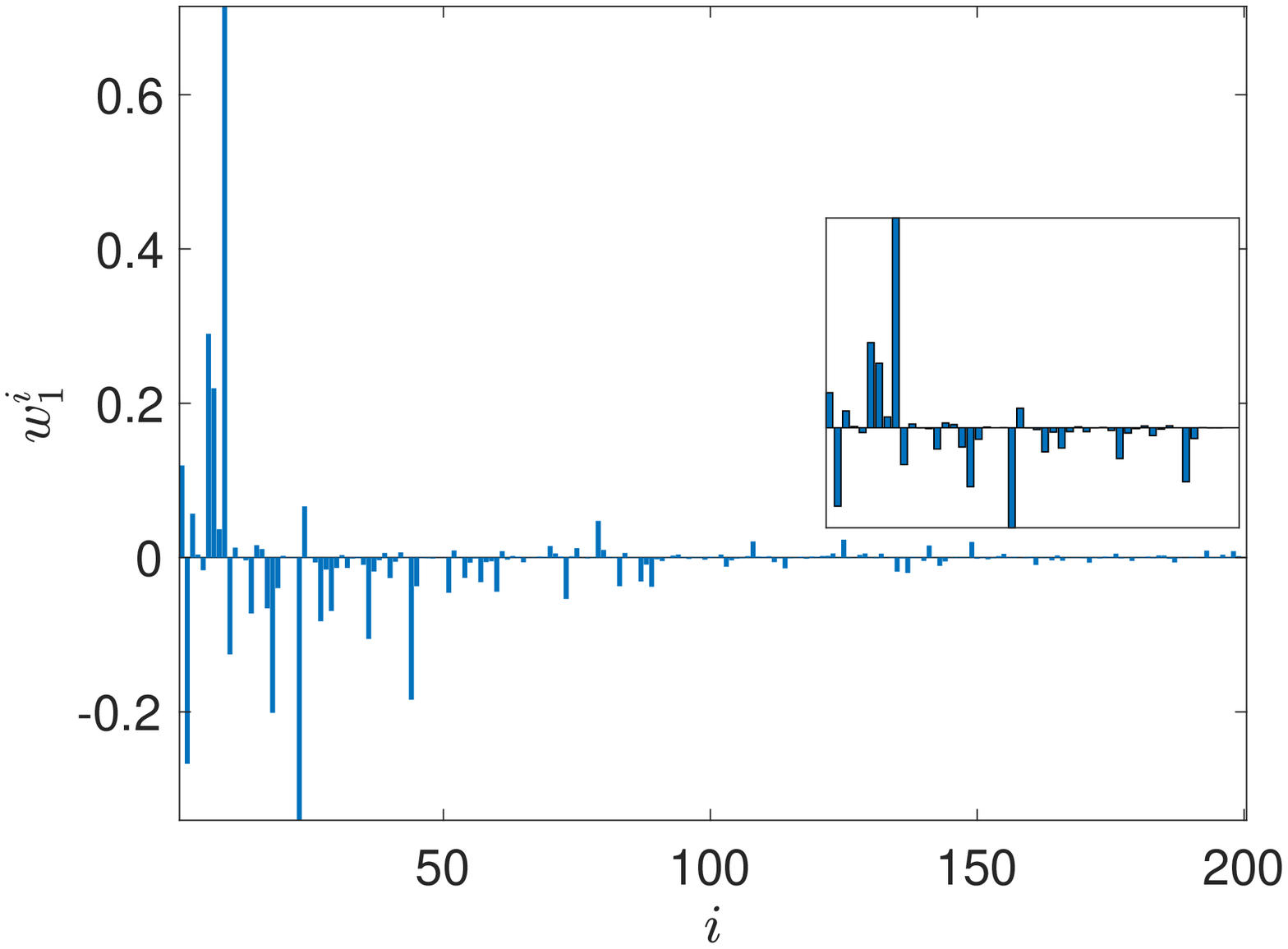} 
        \caption{Components of the dominant eigenvector $\vec{w}$ for the output KL modes, 
        $p_1$, $p_2$, and $p_3$ (left, middle, and right images, respectively). The inset plot
        shows the first 50 components in each case.}
        \label{fig:evecs}
\end{figure}

Next, we compute the surrogate models for the output KL  modes, 
following the strategy described in Section~\ref{sec:method}.
The surrogate models $\tilde{G}_k(\mat{W}_{k,1}^T \vec\xi) 
\approx p_k(\vec\xi)$, $k = 1, \ldots, N$ are
constructed by regression fit. Specifically, letting $p_{i,k}=p_k(\vec\xi_i)$
and $y_i = \mat{W}_{k,1}^T \vec\xi$, 
$i = 1, \ldots, N_s$, we compute least squares approximating polynomials
$\tilde{G}_k$ that minimize 
\[
\sum_{i=1}^{N_s}(\tilde{G}_k(y_i)-p_{i,k})^2, \quad k=1,\ldots,N.
\]
In our computations, we found that linear regressions fits were suitable for
the dominant output KL modes, except modes $1$, $11$, and $12$, for which we
used a quadratic fit. For illustration, we report the computed surrogate models
for the first three output modes in Figure~\ref{fig:eigs}~(bottom) and for the
mode $p_{12}$ in  Figure~\ref{fig:mode12}~(right).
We now have all the
pieces to form the surrogate model for the pressure field
$p(\vec{x},\vec{\xi})$: 
\begin{equation}\label{equ:distributed_surrogate}
p(\vec{x},\vec{\xi}) \approx \hat{p}(\vec{x},\vec{\xi}) = \bar{p}(\vec{x}) + \sum_{k=1}^{N} \sqrt{\lambda_k}\tilde{G}_k(\mat{W}_{k,1}^T\vec{\xi})\phi_k(\vec{x}).
\end{equation}
Below, we examine the accuracy of the computed surrogate model and 
show its effectiveness in capturing the statistical properties of the 
pressure field.

\subsection{The accuracy of the surrogate model}\label{sec:AS_Bio_Test}
In this section, we provide various tests of accuracy that examine different
aspects of the proposed method. To provide a baseline for comparision, 
we computed $10,000$ realizations of the exact pressure field and its 
surrogate model approximation. 

We begin by examining the success of the low-rank KL approximation of the 
pressure field in capturing the variance of the process. The variance of the pressure field
can be obtained from 
$\text{Var}(p(\vec{x},\vec{\xi})) = \mathbb{E}\{p(\vec{x},\vec{\xi})^2\}-\mathbb{E}\{p(\vec{x},\vec{\xi})\}^2$,
which we approximate using the computed samples of $p(\vec{x}, \vec\xi)$. The variance field 
for the truncated KL exansion of $p$ is determined completely by the spectral decomposition of
its (approximate) covariance operator:
\[
\text{Var}(\hat{p}(\vec{x},\vec{\xi})) = \sum_{k=1}^{N}\lambda_k\phi_k(\vec{x})^2,
\]
where $\lambda_k$ and
$\phi_k(\vec{x})$ are the eigenvalues and eigenvectors of the covariance
operator $C_p$.  Taking the square root we obtain the standard deviation of
both the exact pressure field and its approximation; results are shown in
Figure \ref{fig:std}. We note that for both, the standard deviation is
highest at the center of the tumor and decreases to zero as we move away from
the center. We also observe that 
even a low-rank approximation to $p(\vec{x}, \vec\xi)$ captures the standard 
deviation of the pressure field well. 
\begin{figure}[h!]
	\begin{center}
		\begin{tabular}{cccc}
			\includegraphics[width=50mm]{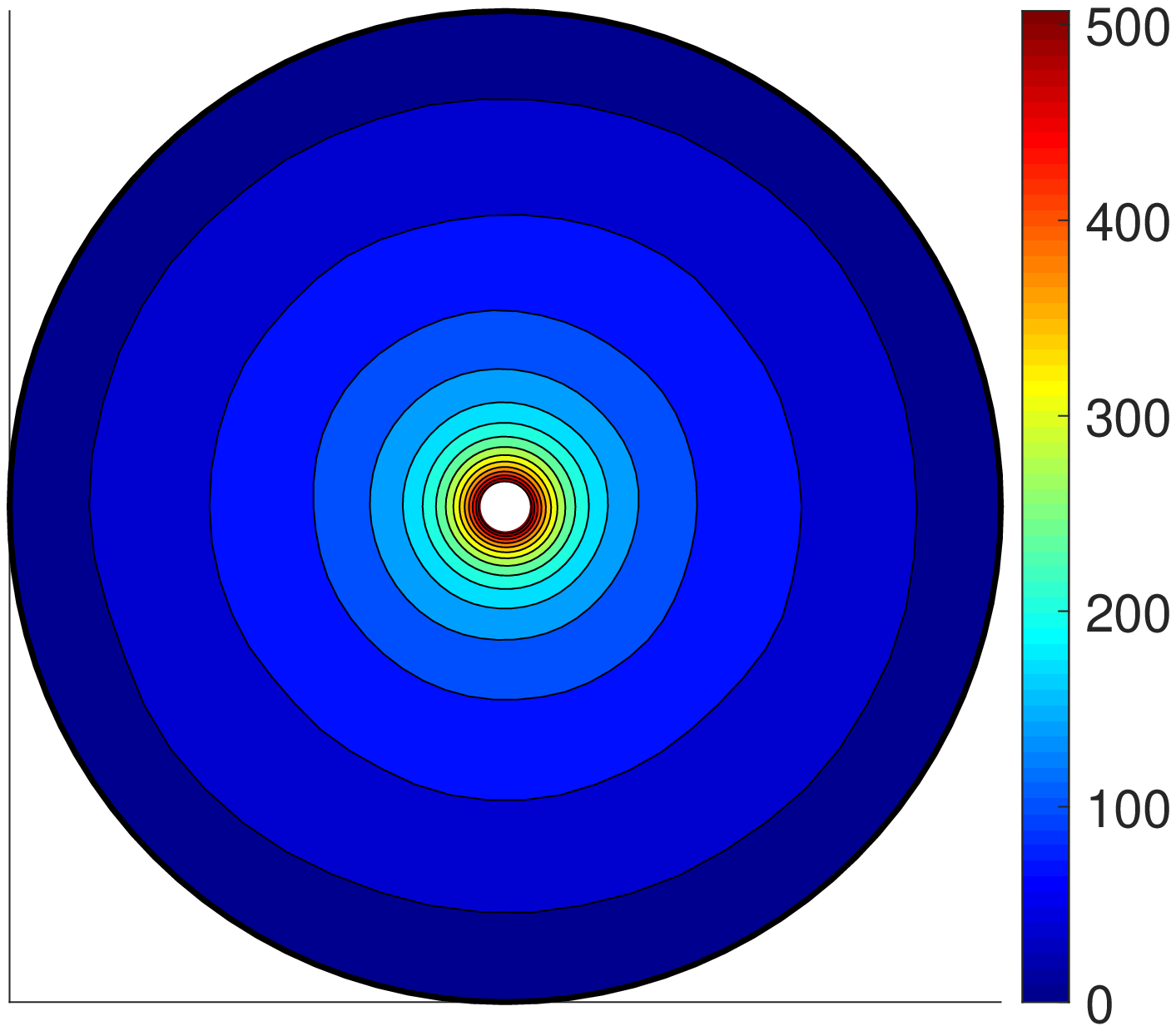} &   
			\includegraphics[width=50mm]{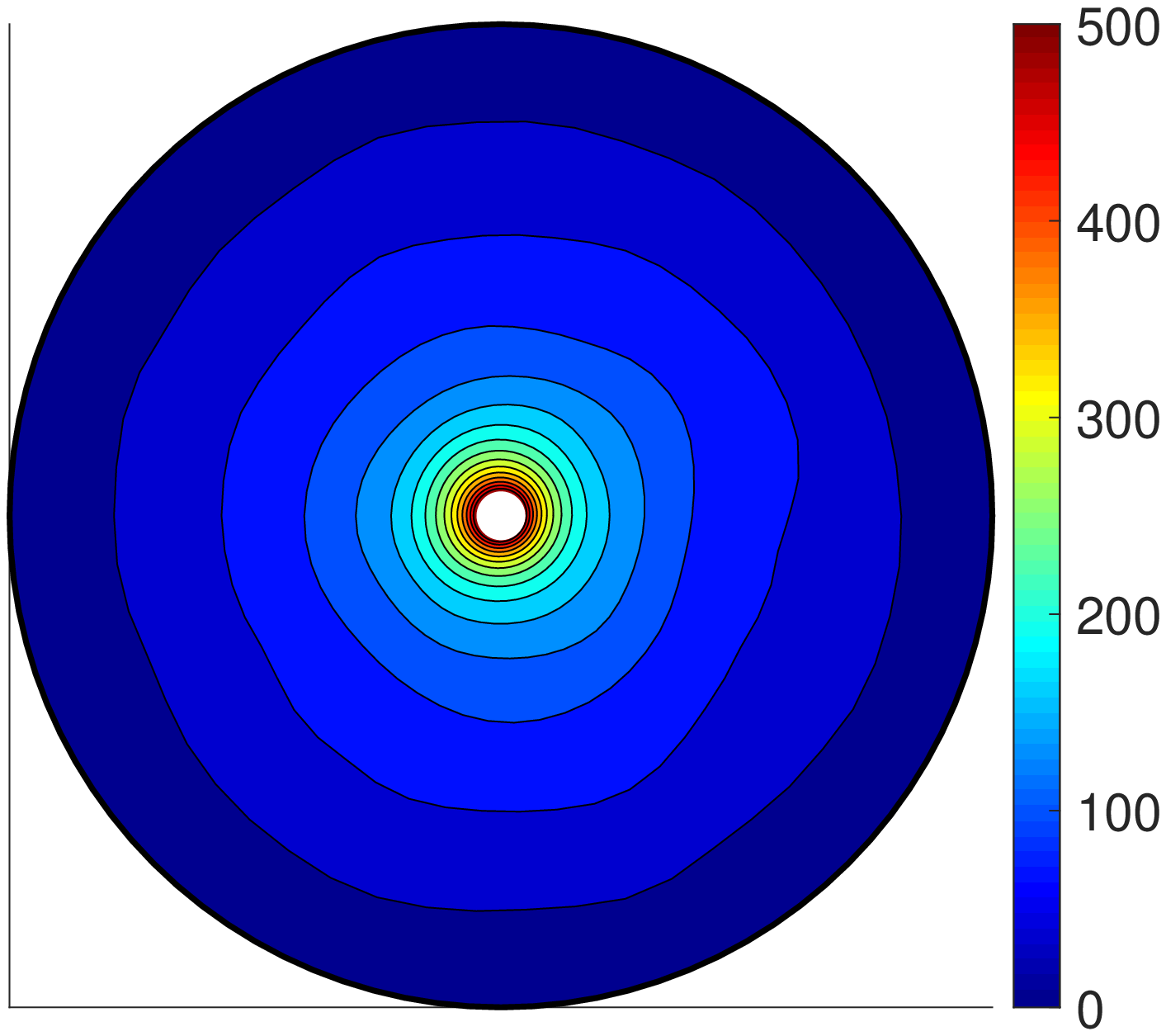} &
			\includegraphics[width=50mm]{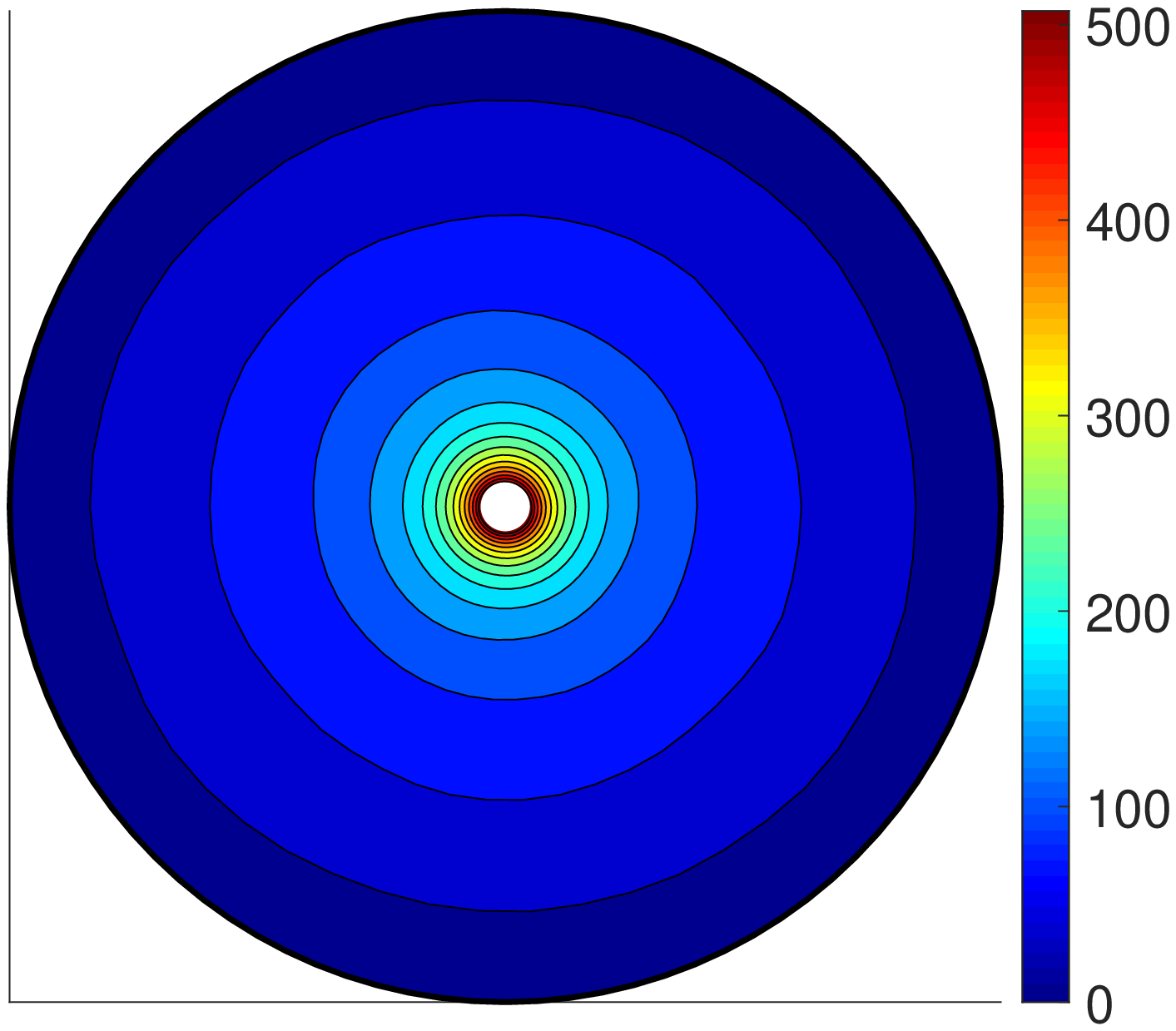} \\
		\end{tabular}
	\end{center}
	\caption{(Left) standard deviation of $p(\vec{x},\vec{\xi})$ using sampling, (middle) standard deviation of $p(\vec{x},\vec{\xi})$ using KLE (15 modes), and (right) standard deviation of $p(\vec{x},\vec{\xi})$ using KLE (50 modes).}
	\label{fig:std}
\end{figure}

Next, we study the accuracy of the computed surrogate model by focusing on 
the relative error indicator 
\[
    E_{rel}(\vec\xi) = 
    \frac{\int_\D | p(\vec{x}, \vec\xi) - \hat{p}(\vec{x}, \vec\xi)|\, d\vec{x}}
         {\int_\D | p(\vec{x}, \vec\xi)\, d\vec{x}}, 
\]
where $\hat{p}(\vec{x}, \vec\xi)$ is computed according
to~\eqref{equ:distributed_surrogate}.  In Figure~\ref{fig:error_dist}~(left),
we plot the expected value of $E_{rel}$ as the number of output KL modes increases; the
dashed lines indicate the fifth and ninty fifth percentiles. We note that with
$N = 15$ output KL modes, the relative error is about $5$\% on
average.  To better understand the behavior of the relative error, we report
its distribution in Figure~\ref{fig:error_dist}~(right), where we used $N = 15$
output KL modes in computing $\hat{p}$.  We also show the distribution of the
relative error computed over a subdomin $\D' = \{ \vec{x} \in \D : \| \vec{x}\|
\leq 2\}$ in Figure~\ref{fig:error_dist}~(right).  This is done to quantify the
approximation errors near the injection site. We observe that the distribution of the
error is shifted to the left, indicating smaller approximation errors in $\mathcal{D}'$
(with high probability).
To see more clearly how the
distribution of the relative error evolves as the number of KL modes
increase, we show the probability density 
function of $E_{rel}$ corresponding to different number of output KL modes in Figure~\ref{fig:error_pdf_surf}.
Note that not only does the mode of the distribution get smaller, the spread of the
error also decreases, which can be inferred from Figure~\ref{fig:error_dist}~(left) as well.
\begin{figure}[ht]\centering
\includegraphics[width=.4\textwidth]{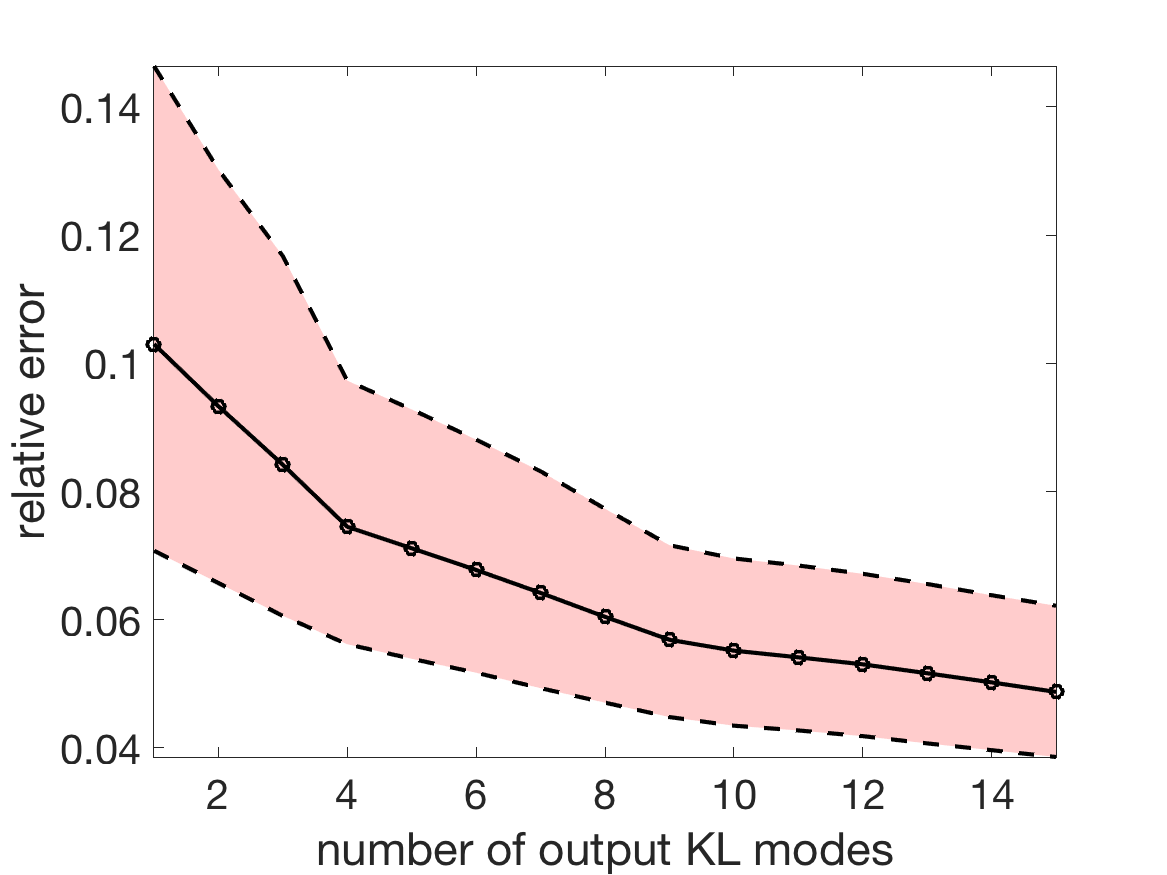}
\includegraphics[width=.4\textwidth]{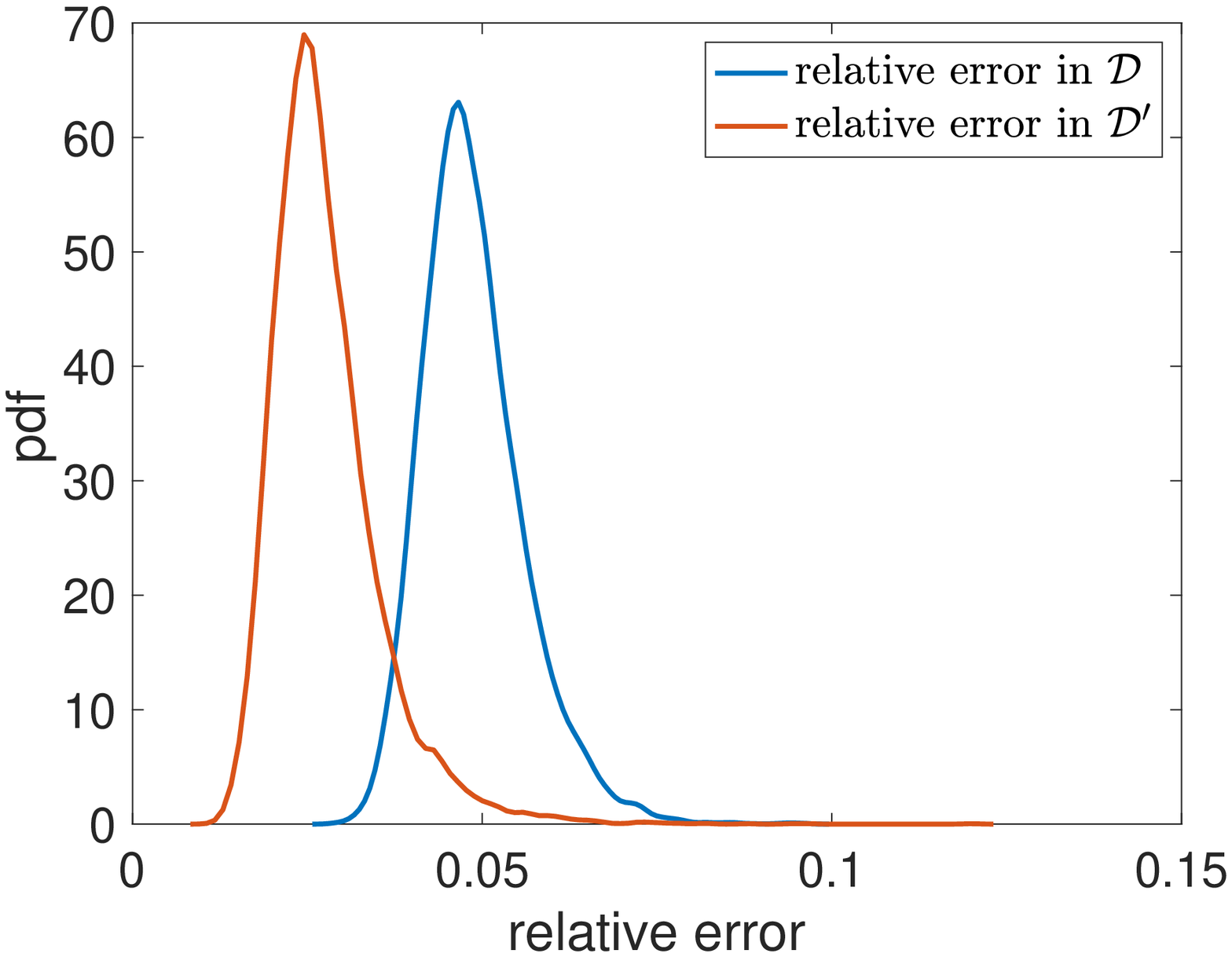}
\caption{Left: The expected value of the error $E_{rel}$ (solid black line)
along with the fifth and the ninty fifth percentiles (dashed lines) of 
$E_{rel}$. Right: Distribution of $E_{rel}$ computed over the entire domain $\D$ (blue)
and over the subdomian $\D'$ (red).} 
\label{fig:error_dist}
\end{figure}
\begin{figure}[ht]\centering
\includegraphics[width=.5\textwidth]{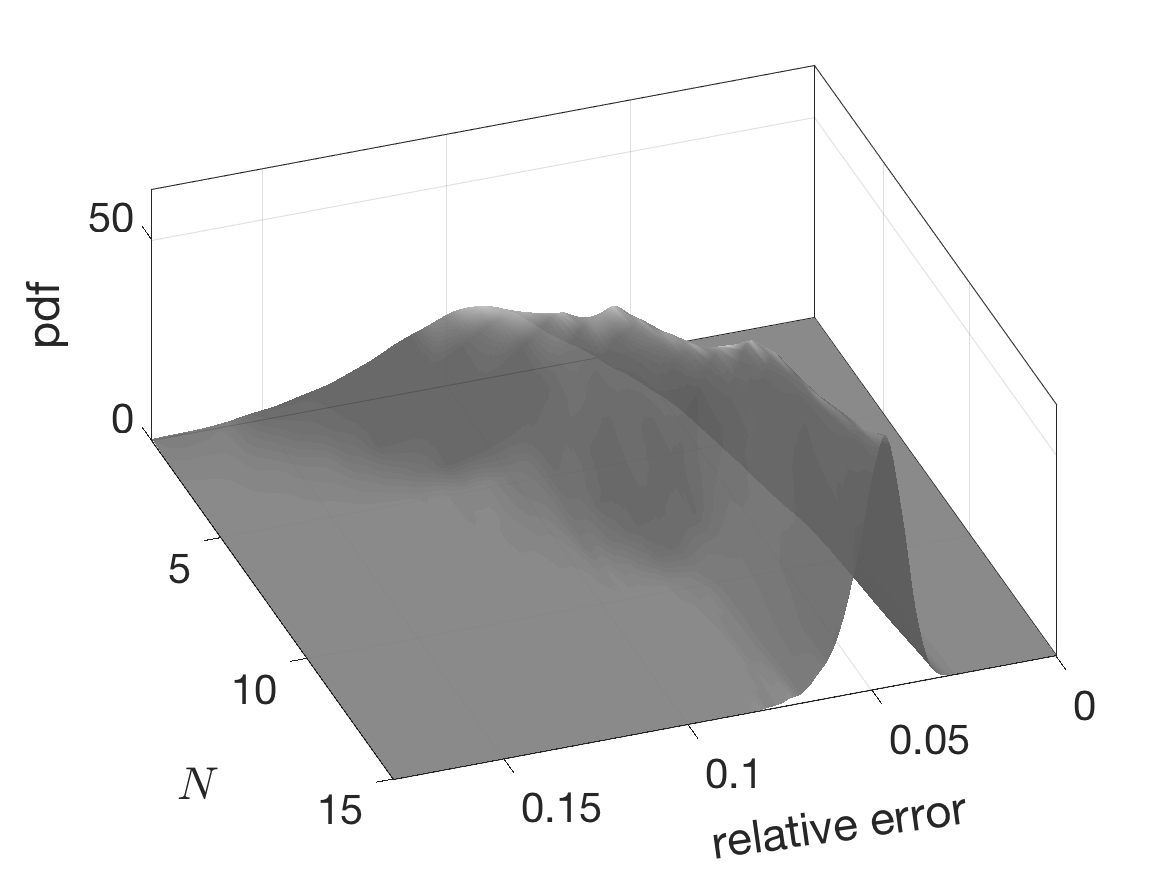}
\caption{Distribution of the error $E_{rel}$ as the number of 
output KL modes increase from $N = 1$ to $N = 15$.}
\label{fig:error_pdf_surf}
\end{figure}

In Figure~\ref{fig:pdfs}, we report the distribution of the average pressure
along concentric circles of various radii.  This shows that the computed
surrogate captures statistical properties of the pressure well in
different parts of the domain.
\begin{figure}[ht]
\includegraphics[width=.3\textwidth]{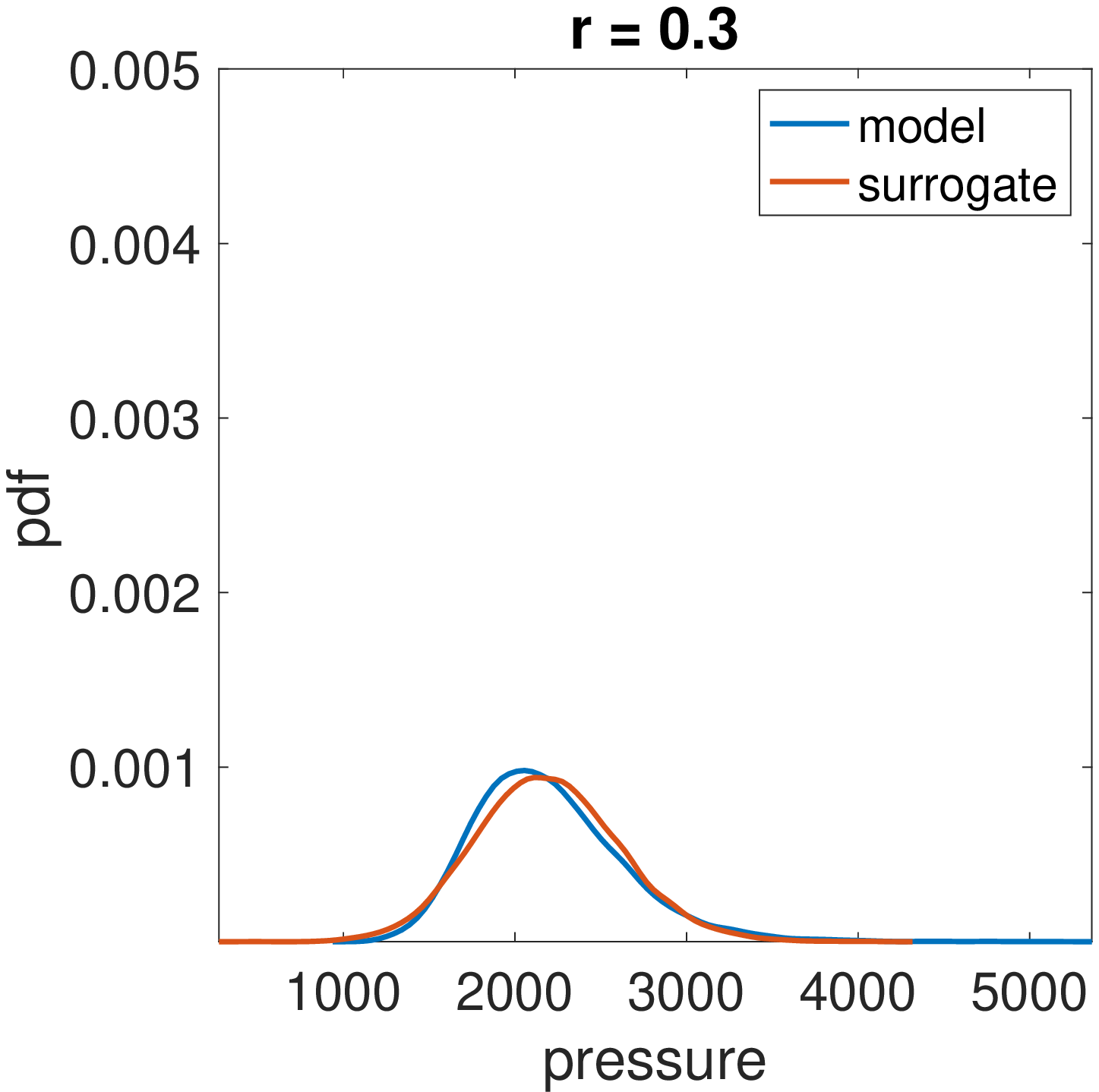}
\includegraphics[width=.3\textwidth]{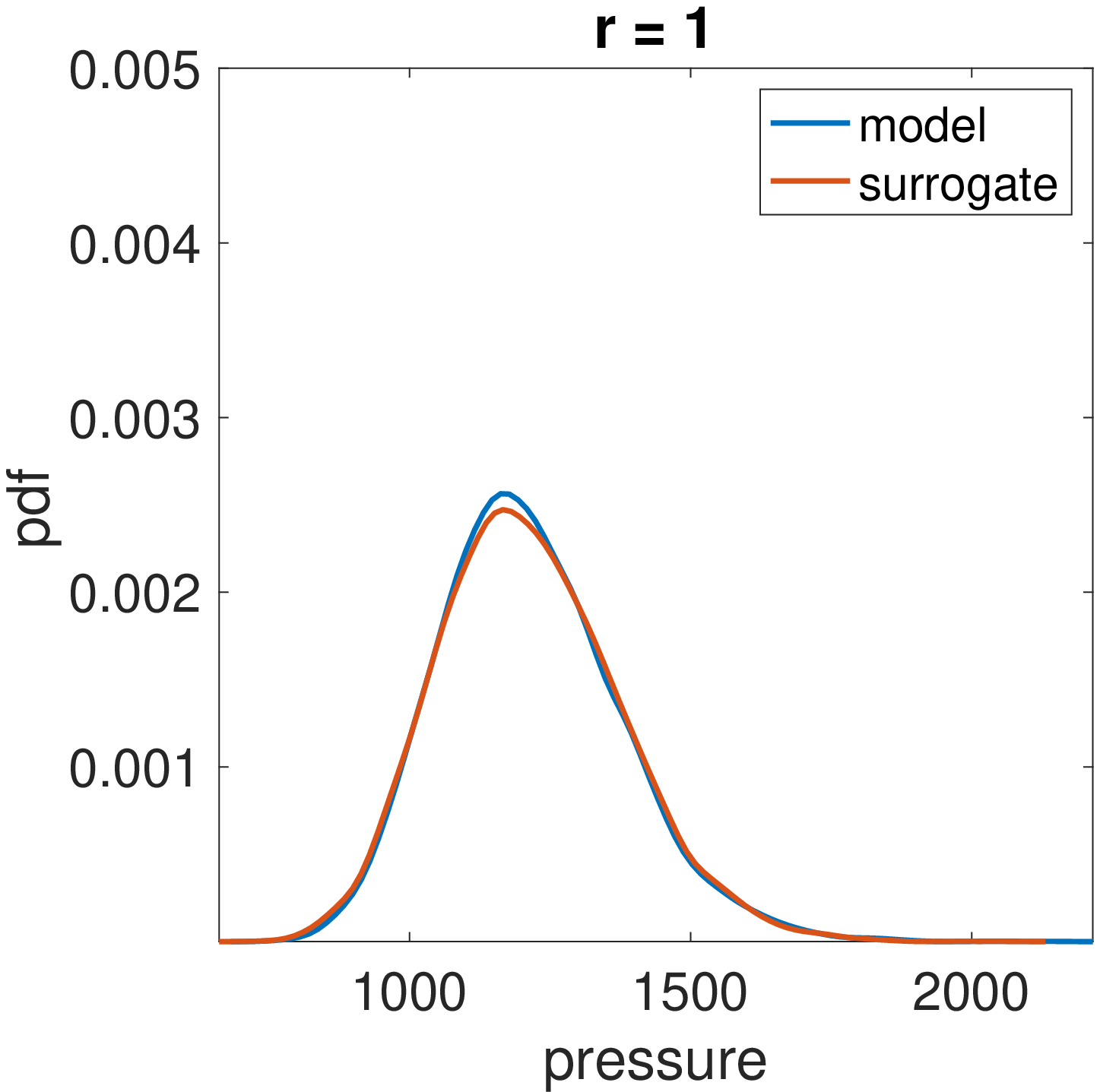}
\includegraphics[width=.3\textwidth]{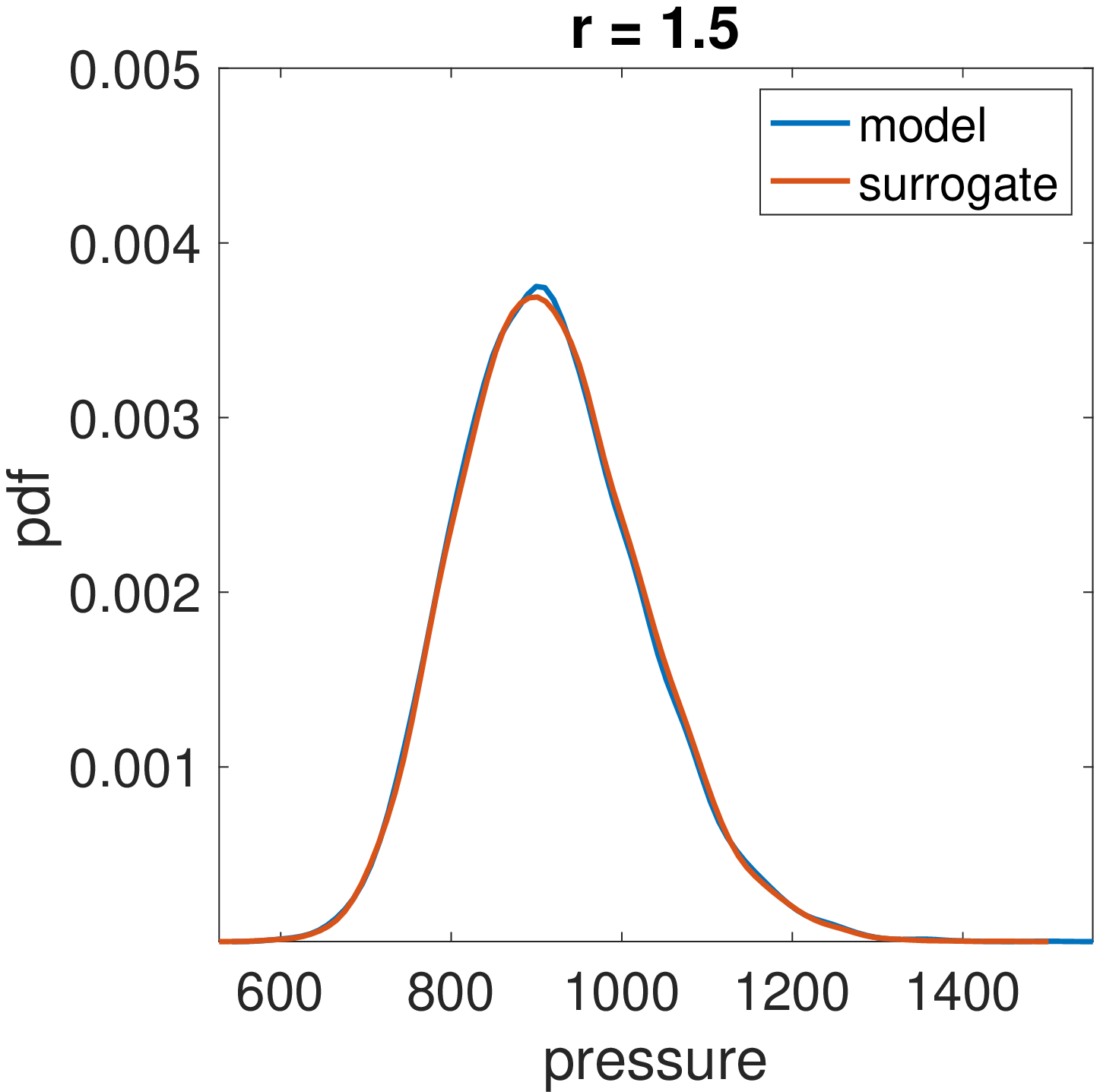}
\\
\includegraphics[width=.3\textwidth]{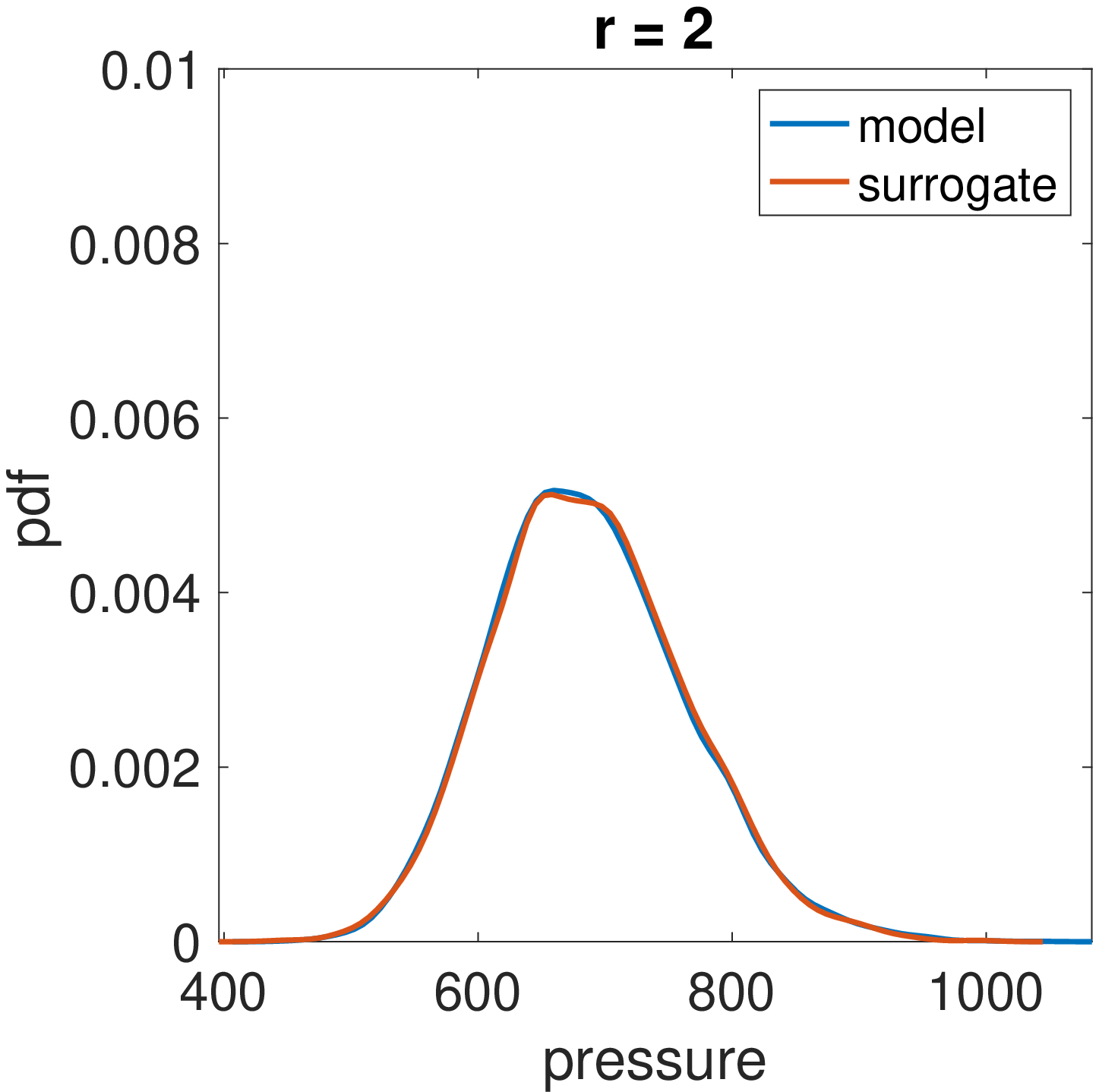}
\includegraphics[width=.3\textwidth]{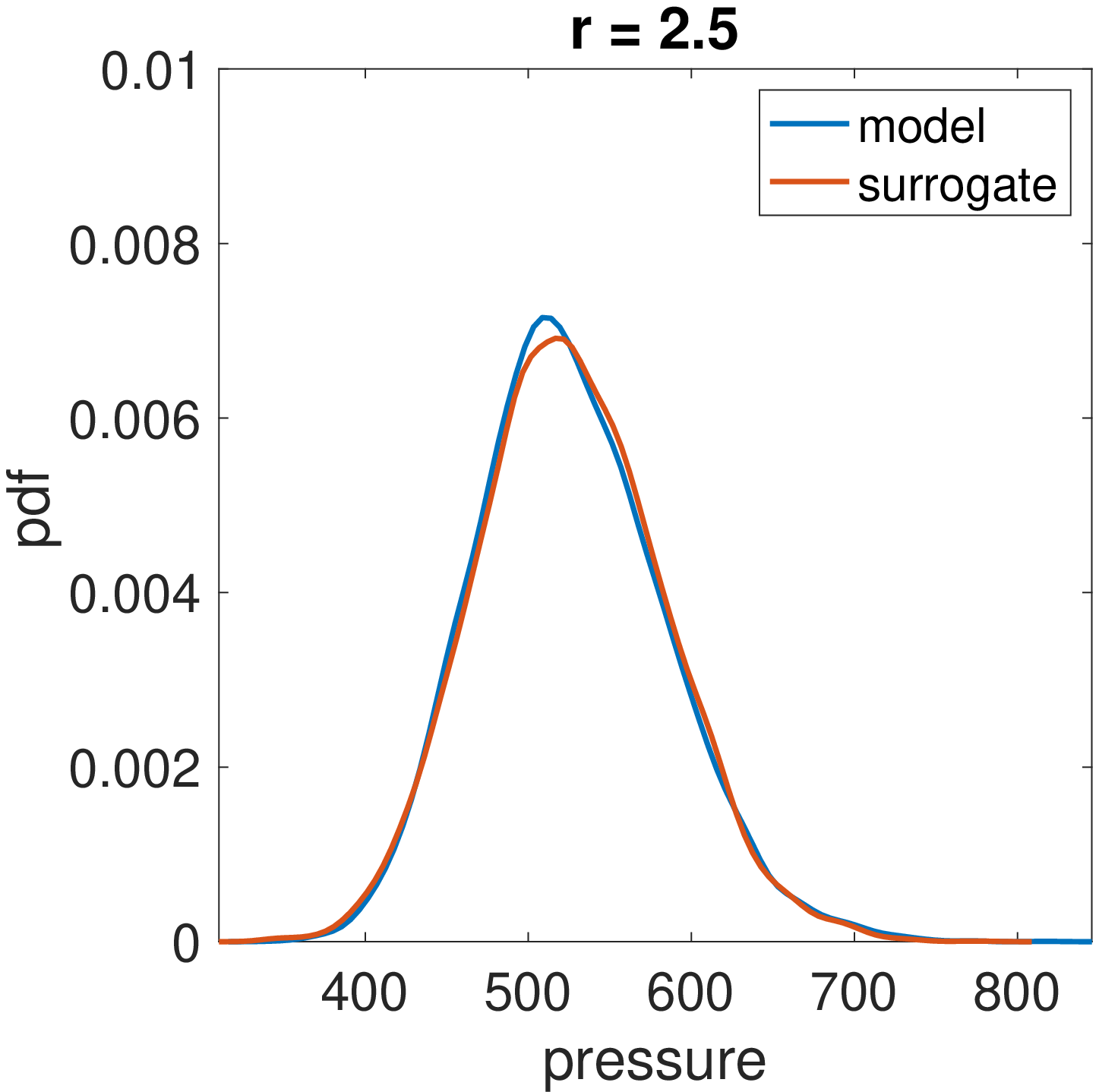}
\includegraphics[width=.3\textwidth]{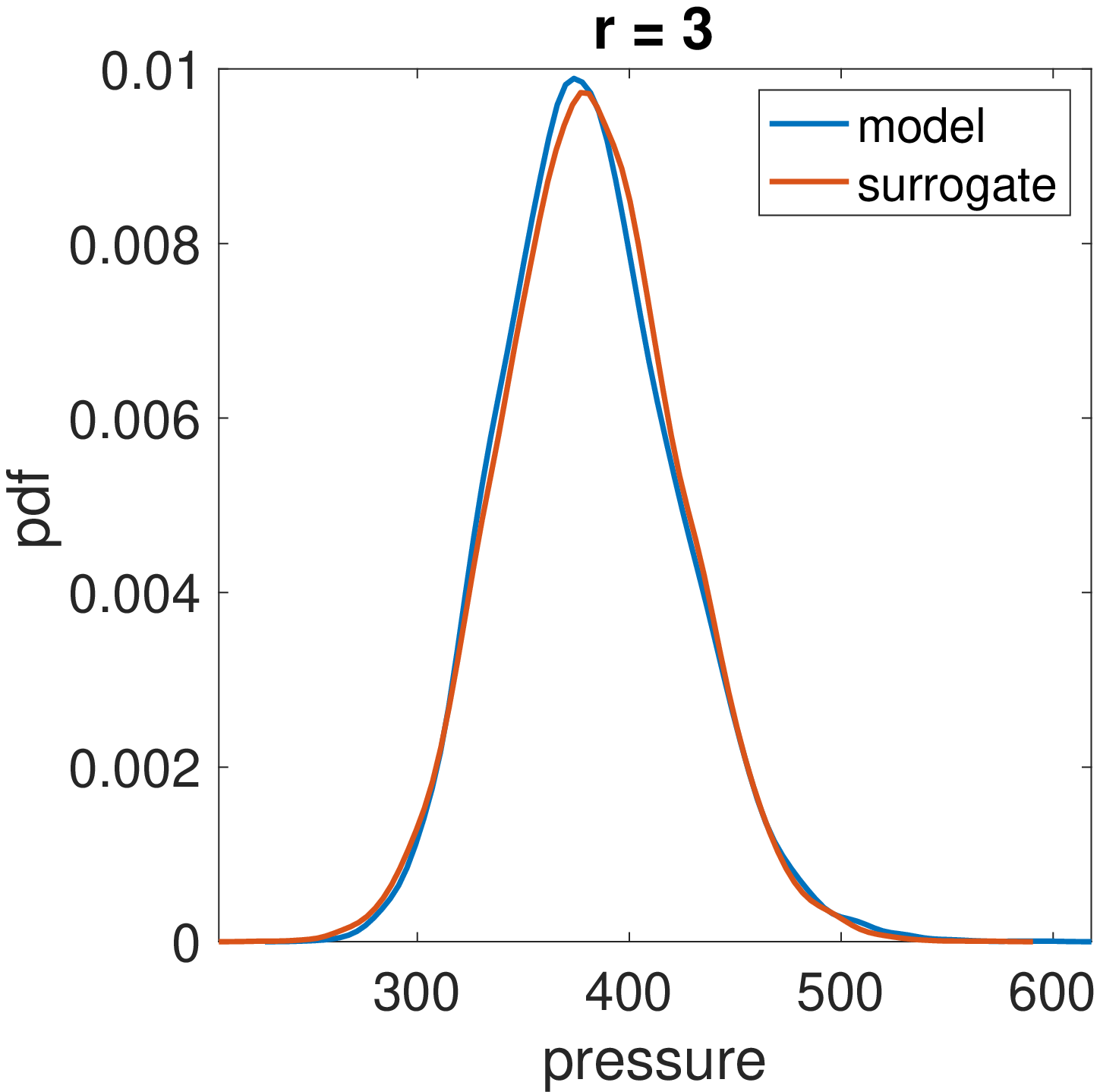}
\caption{Distribution of the pressure averaged over concentric circles of 
various radii.}
\label{fig:pdfs}
\end{figure}

\section{Concluding remarks}
\label{sec:conc}
We have presented a distributed active subspace method for scalable surrogate
modeling of PDE-governed physical processes with high-dimensional inputs and
function valued outputs. To save the modeling efforts spent on function valued
outputs, we employ the truncated KL expansion to decouple the spatial
dimensions with those of the random variables. As a result, the randomness in
outputs is fully represented by scalar valued KL modes. For elliptic PDEs, as
observed in the work, a low-rank KL representation of the model output is usually
sufficient for an accurate representation. To reduce the dimension of
inputs, we construct active subspaces for each of the
dominant output KL modes. Since output gradients with respect to the random
variables need to be calculated when constructing active subspaces, we develop
an adjoint-based framework to ensure that the computational cost
does not scale with the input dimension. The method development is complemented
by a rigorous mathematical formulation as well as theoretical analysis of 
errors due to active subspace projection and output KL 
representation.

We then deploy the distributed active subspace method to conduct surrogate
modeling of the pressure field in a biotransport model in tumors with an
uncertain permeability field.  We demonstrate that a low-rank representation of
the pressure field (i.e., output) can be achieved with the truncated KL
expansion.  The input (i.e., random variables used to represent the uncertain
log-permeability field) dimension can be very high (e.g., several hundreds)
when the correlation length of the log-permeability field is small. However, we
observe that dominant output KL modes admit low (one or two) dimensional active
subspaces.  We observe that the average relative error between the surrogate
model and the PDE solution can  be controlled under 5\% even when a very
low-rank representation of the outputs is employed. We also show that the
surrogate models can capture statistical properties of the pressure well in
different parts of the domain.      

In future work, we will investigate extensions of the distributed active
subspace method to surrogate modeling of the time-dependent diffusion and
convection-diffusion processes in biological and geological flows. We envision
that the truncated KL expansion can be used to extract spatiotemporal coherent
structures from the physical process, and the low-dimensional active subspaces
can then be constructed for the scalar valued KL modes associated with these
structures. This surrogate modeling approach will contribute to cost-effective
forward uncertainty quantification, and facilitate solution of inverse problems
under uncertainty, such as Bayesian inversion of tissue material properties
from medical images.

\bibliographystyle{spmpsci}
\bibliography{refs}

\newcommand{\noop}[1]{}
\begin{thebibliography}{10}
\providecommand{\url}[1]{{#1}}
\providecommand{\urlprefix}{URL }
\expandafter\ifx\csname urlstyle\endcsname\relax
  \providecommand{\doi}[1]{DOI~\discretionary{}{}{}#1}\else
  \providecommand{\doi}{DOI~\discretionary{}{}{}\begingroup
  \urlstyle{rm}\Url}\fi

\bibitem{AlexanderianGremaudSmith17}
Alexanderian, A., Gremaud, P., Smith, R.: Variance-based sensitivity analysis
  for time-dependent processes.
\newblock In review (arXiv: \url{https://arxiv.org/abs/1711.08030})  (2019)

\bibitem{ARSY2019}
{Alexanderian}, A., {Reese}, W., {Smith}, R.C., {Yu}, M.: {Model input and
  output dimension reduction using {K}arhunen--{L}o\`{e}ve expansions with
  application to biotransport}.
\newblock ASCE-ASME Journal of Risk and Uncertainty in Engineering Systems Part
  B: Mechanical Engineering \textbf{Accepted} (2019).
\newblock \url{https://ui.adsabs.harvard.edu/\#abs/2019arXiv190306314A}

\bibitem{AlexanderianZhuSalloumEtAl17}
Alexanderian, A., Zhu, L., Salloum, M., Ma, R., Yu, M.: Investigation of
  biotransport in a tumor with uncertain material properties using a
  non-intrusive spectral uncertainty quantification method.
\newblock J. Biomech. Eng. pp. 091006--1--091006--11 (2017)

\bibitem{BabuvskaNobileTempone07}
Babu{\v{s}}ka, I., Nobile, F., Tempone, R.: A stochastic collocation method for
  elliptic partial differential equations with random input data.
\newblock SIAM Journal on Numerical Analysis \textbf{45}(3), 1005--1034 (2007)

\bibitem{Clark:91}
Clark, W.H.: {Tumour progression and the nature of cancer}.
\newblock Br J Cancer \textbf{64}, 631--44 (1991)

\bibitem{Chen:07}
Clark, W.H.: {Biphasic finite element model of solute transport for direct
  infusion into nervous tissue}.
\newblock Annals of Biomedical Engineering \textbf{35}, 2145–--2158 (2007)

\bibitem{CleavesAlexanderianGuyEtAl19}
{Cleaves}, H.L., {Alexanderian}, A., {Guy}, H., {Smith}, R.C., {Yu}, M.:
  {Derivative-based global sensitivity analysis for models with
  high-dimensional inputs and functional outputs}.
\newblock arXiv e-prints arXiv:1902.04630 (2019)

\bibitem{Constantine15}
Constantine, P.: Active Subspaces: Emerging Ideas in Dimension Reduction for
  Parameter Studies.
\newblock SIAM, Philadelphia (2015)

\bibitem{ConstantineDiaz17}
Constantine, P.G., Diaz, P.: Global sensitivity metrics from active subspaces.
\newblock Reliability Engineering \& System Safety \textbf{162}, 1--13 (2017)

\bibitem{ConstantineBattery17}
Constantine, P.G., Doostan, A.: Time‐dependent global sensitivity analysis
  with active subspaces for a lithium ion battery model.
\newblock Statistical Analysis and Data Mining: The ASA Data Science Journal
  \textbf{10}, 243--262 (2017)

\bibitem{constantine2014}
Constantine, P.G., Dow, E., Wang, Q.: Active subspace methods in theory and
  practice: applications to kriging surfaces.
\newblock SIAM Journal on Scientific Computing \textbf{36}(4), A1500--A1524
  (2014)

\bibitem{Scramjet15}
Constantine, P.G., Emory, M., Larsson, J., Iaccarino, G.: Exploiting active
  subspaces to quantify uncertainty in the numerical simulation of the hyshot
  ii scramjet.
\newblock Journal of Computational Physics \textbf{302}, 1--20 (2015)

\bibitem{Deb:09CPD}
Debbage, P.: Targeted drugs and nanomedicine: present and future.
\newblock Current Pharmaceutical Design \textbf{15}, 153--72 (2009)

\bibitem{Doostan07}
Doostan, A., Ghanem, R.G., Red-Horse, J.: Stochastic model reduction for chaos
  representations.
\newblock Computer Methods in Applied Mechanics and Engineering
  \textbf{196}(37-40), 3951--3966 (2007)

\bibitem{Elman17}
Elman, H.: Solution algorithms for stochastic galerkin discretizations of
  differential equations with random data.
\newblock Handbook of Uncertainty Quantification pp. 1--16 (2017)

\bibitem{friedman93}
{Friedman}, J.: Fast {MARS}.
\newblock Tech. Rep. 110, Laboratory for Computational Statistics, Department
  of Statistics, Stanford University (1993)

\bibitem{GamboaJanonKleinEtAl14}
Gamboa, F., Janon, A., Klein, T., Lagnoux, A., et~al.: Sensitivity analysis for
  multidimensional and functional outputs.
\newblock Electronic Journal of Statistics \textbf{8}(1), 575--603 (2014)

\bibitem{Ghanem98}
Ghanem, R.: Probabilistic characterization of transport in heterogeneous media.
\newblock Computer Methods in Applied Mechanics and Engineering
  \textbf{158}(3), 199 -- 220 (1998).
\newblock \doi{https://doi.org/10.1016/S0045-7825(97)00250-8}.
\newblock
  \urlprefix\url{http://www.sciencedirect.com/science/article/pii/S0045782597002508}

\bibitem{GhanemSpanos90}
Ghanem, R.G., Spanos, P.D.: Stochastic finite elements: a spectral approach.
\newblock Springer-Verlag, New York (1991).
\newblock \doi{10.1007/978-1-4612-3094-6}.
\newblock \urlprefix\url{http://dx.doi.org/10.1007/978-1-4612-3094-6}

\bibitem{Graham15KuoKuoNicholsEtAl15}
Graham, I.G., Kuo, F.Y., Nichols, J.A., Scheichl, R., Schwab, C., Sloan, I.H.:
  Quasi-monte carlo finite element methods for elliptic pdes with lognormal
  random coefficients.
\newblock Numerische Mathematik \textbf{131}(2), 329--368 (2015)

\bibitem{Gunzburger03}
Gunzburger, M.: Perspectives in flow control and optimization, vol.~5.
\newblock Siam (2003)

\bibitem{HsingEubank15}
Hsing, T., Eubank, R.: Theoretical foundations of functional data analysis,
  with an introduction to linear operators.
\newblock John Wiley \& Sons (2015)

\bibitem{IoossSaltelli17}
Iooss, B., Saltelli, A.: Introduction to sensitivity analysis.
\newblock In: R.~Ghanem, D.~Higdon, H.~Owhadi (eds.) Handbook of uncertainty
  quantification, pp. 1103--1122. Springer (2017)

\bibitem{ConstantineGeo15}
Jefferson, J., Gilbert, J., Constantine, P., Maxwell, R.: Active subspaces for
  sensitivity analysis and dimension reduction of an integrated hydrologic
  model.
\newblock Computers \& Geosciences \textbf{83}, 127--138 (2015)

\bibitem{ji2018shared}
Ji, W., Wang, J., Zahm, O., Marzouk, Y.M., Yang, B., Ren, Z., Law, C.K.: Shared
  low-dimensional subspaces for propagating kinetic uncertainty to multiple
  outputs.
\newblock Combustion and Flame \textbf{190}, 146--157 (2018)

\bibitem{KucherenkoIooss17}
Kucherenko, S., Iooss, B.: Derivative-based global sensitivity measures.
\newblock In: R.~Ghanem, D.~Higdon, H.~Owhadi (eds.) Handbook of Uncertainty
  Quantification. Springer (2017)

\bibitem{LeMaitreKnio04}
Le~Ma{\i}tre, O., Knio, O., Najm, H., Ghanem, R.: Uncertainty propagation using
  wiener--haar expansions.
\newblock Journal of computational Physics \textbf{197}(1), 28--57 (2004)

\bibitem{LeMaitreKnio10}
Le~Ma{\^{\i}}tre, O.P., Knio, O.M.: Spectral methods for uncertainty
  quantification.
\newblock Scientific Computation. Springer, New York (2010).
\newblock \doi{10.1007/978-90-481-3520-2}.
\newblock \urlprefix\url{http://dx.doi.org/10.1007/978-90-481-3520-2}.
\newblock With applications to computational fluid dynamics

\bibitem{LeMaitreReaganNajmEtAl02}
Le~Ma{\^i}tre, O.P., Reagan, M.T., Najm, H.N., Ghanem, R.G., Knio, O.M.: A
  stochastic projection method for fluid flow: Ii. random process.
\newblock Journal of computational Physics \textbf{181}(1), 9--44 (2002)

\bibitem{Loeve77}
Lo{\`e}ve, M.: Probability theory. {I}, fourth edn.
\newblock Springer-Verlag, New York-Heidelberg (1977).
\newblock Graduate Texts in Mathematics, Vol. 45

\bibitem{ShapeOpt14}
Lukaczyk, T.W., Palacios, F., Alonso, J.J., Constantine, P.: Active subspaces
  for shape optimization.
\newblock In: the 10th AIAA Multidisciplinary Design Optimization Conference.
  National Harbor, Maryland (2014).
\newblock AIAA-2014-1171

\bibitem{ma:12TF}
Ma, R., Su, D., Zhu, L.: {Multiscale simulation of nanopartical transport in
  deformable tissue during an infusion process in hyperthermia treatments of
  cancers.}
\newblock In: W.J. Minkowycz, E.~Sparrow, J.P. Abraham (eds.) Nanoparticle Heat
  Transfer and Fluid Flow, Computational \& Physical Processes in Mechanics \&
  Thermal Science Series, vol.~4. CRC Press, Taylor \& Francis Group (2012)

\bibitem{Frontier:16}
Mangado, N., Piella, G., Noailly, J., Pons-Prats, J., Ángel
  González~Ballester, M.: Analysis of uncertainty and variability in finite
  element computational models for biomedical engineering: Characterization and
  propagation.
\newblock Front Bioeng Biotechnol. \textbf{4}, 85 (2016)

\bibitem{MatthiesKeese05}
Matthies, H.G., Keese, A.: Galerkin methods for linear and nonlinear elliptic
  stochastic partial differential equations.
\newblock Computer methods in applied mechanics and engineering
  \textbf{194}(12-16), 1295--1331 (2005)

\bibitem{PrieurTarantola17}
Prieur, C., Tarantola, S.: Variance-based sensitivity analysis: Theory and
  estimation algorithms.
\newblock In: R.~Ghanem, D.~Higdon, H.~Owhadi (eds.) Handbook of Uncertainty
  Quantification, pp. 1217--1239. Springer (2017)

\bibitem{RasmussenWilliams06}
Rasmussen, C.E., Williams, C.: Gaussian processes for machine learning
  cambridge (2006)

\bibitem{Russi10}
Russi, T.M.: Uncertainty quantification with experimental data and complex
  system models.
\newblock Ph.D. thesis, University of California, Berkeley (2010)

\bibitem{SaadGhanem09}
Saad, G., Ghanem, R.: Characterization of reservoir simulation models using a
  polynomial chaos-based ensemble kalman filter.
\newblock Water Resources Research \textbf{45}(4) (2009)

\bibitem{maher:08}
Salloum, M., Ma, R., Weeks, D., Zhu, L.: {Controlling nanoparticle delivery in
  magnetic nanoparticle hyperthermia for cancer treatment: experimental study
  in agarose gel}.
\newblock Int. J. Hyperthermia \textbf{24}, 337--345 (2008)

\bibitem{Sobol:1990}
Sobol, I.: Estimation of the sensitivity of nonlinear mathematical models.
\newblock Matematicheskoe Modelirovanie \textbf{2}(1), 112--118 (1990)

\bibitem{Sobol:2001}
Sobol, I.: Global sensitivity indices for nonlinear mathematical models and
  their monte carlo estimates.
\newblock Mathematics and Computers in Simulation \textbf{55}(1-3), 271 -- 280
  (2001)

\bibitem{SobolKucherenko09}
Sobol', I., Kucherenko, S.: Derivative based global sensitivity measures and
  their link with global sensitivity indices.
\newblock Mathematics and Computers in Simulation \textbf{79}(10), 3009--3017
  (2009)

\bibitem{XiuKarniadakis03}
Xiu, D., Karniadakis, G.E.: Modeling uncertainty in flow simulations via
  generalized polynomial chaos.
\newblock Journal of Computational Physics \textbf{187}(1), 137 -- 167 (2003).
\newblock \doi{https://doi.org/10.1016/S0021-9991(03)00092-5}.
\newblock
  \urlprefix\url{http://www.sciencedirect.com/science/article/pii/S0021999103000925}

\bibitem{zahm2018gradient}
Zahm, O., Constantine, P., Prieur, C., Marzouk, Y.: Gradient-based dimension
  reduction of multivariate vector-valued functions.
\newblock arXiv preprint arXiv:1801.07922  (2018)

\end{thebibliography}

\appendix
\section{Proofs}\label{sec:proofs}

\begin{myproof}{Theorem}{\ref{thm:fN_error}}
Note that,
\[
e(\vec{x},\vec{\xi}) 
	=\bigg| \sum_{k=1}^N (F_k(\vec{\xi}) -{G}_k(\vec{W}_{k,1}^T\vec{\xi}))\phi_k(\vec{x}) \bigg|
	 \leq \sum_{k=1}^N |F_k(\vec{\xi}) -{G}_k(\vec{W}_{k,1}^T\vec{\xi})||\phi_k(\vec{x})|. 
\]
Then,
\begin{align*}
\mathbb{E}\{\bar{e}\} &= 
\mathbb{E} \bigg\{\int_\X \sum_{k=1}^N |F_k(\vec{\xi}) -{G}_k(\vec{W}_{k,1}^T\vec{\xi})||\phi_k(\vec{x})| dx \bigg\}\\ 
&= \sum_{k=1}^N  \bigg(\int_\X |\phi_k(\vec{x})| dx \bigg) \mathbb{E}\{|F_k(\vec{\xi}) -{G}_k(\vec{W}_{k,1}^T\vec{\xi})| \}
\\
&\leq \displaystyle\sum_{k=1}^N  \left[\int_\X |\phi_k(\vec{x})|^2 dx\right]^\frac{1}{2}\left[\int_\X 1^2 dx\right]^\frac{1}{2} \mathbb{E}\{|F_k(\vec{\xi}) -{G}_k(\vec{W}_{k,1}^T\vec{\xi})^2|\}^{1/2} \leq |\X|^\frac{1}{2} \displaystyle\sum_{k=1}^N \delta_k
\end{align*}
where we have used Cauchy--Schwarz inequlity and Lemma~\ref{lem:basic}(a). 
\end{myproof}

\begin{myproof}{Theorem}{\ref{thm:overall_error}}
We have 
\begin{align*}	
\bar{E}(\vec{\xi}) = \int_\X E(\vec{x},\vec{\xi}) dx &= \int_\X |f(\vec{x},\vec{\xi})-f_N(\vec{x},\vec{\xi})+ f_N(\vec{x},\vec{\xi})-\hat{f}_N(\vec{x},\vec{\xi})| \\
&\leq \int_\X |f(\vec{x},\vec{\xi})-f_N(\vec{x},\vec{\xi})|dx + \int_\X |f_N(\vec{x},\vec{\xi})-\hat{f}_N(\vec{x},\vec{\xi})|dx \\ & \leq \bigg\{\int_\X |f(\vec{x},\vec{\xi})-f_N(\vec{x},\vec{\xi})|^2 dx\bigg\}^{1/2}|\X|^\frac{1}{2} + \bar{e}(\vec{\xi}) \\ & = \bigg\{ \sum_{k=N+1}^\infty F_k(\vec{\xi})^2 \bigg\}^{1/2}|\X|^\frac{1}{2} + \bar{e}(\vec{\xi}).
\end{align*}
Now we consider
\begin{align*}
\mathbb{E}\{\bar{E}(\vec{\xi})\}& 
\leq \mathbb{E}\Bigg(\bigg\{ \sum_{k=N+1}^\infty F_k(\vec{\xi})^2 \bigg\}^{1/2}\Bigg)|\X|^\frac{1}{2} + \mathbb{E}\{\bar{e}(\vec{\xi})\} \\ 
&\leq |\X|^\frac{1}{2} \Bigg[\mathbb{E}\Bigg(\bigg\{\sum_{k=N+1}^\infty F_k(\vec{\xi})^2\bigg\}^{1/2}\Bigg) +\sum_{k=1}^N \delta_k \Bigg] 
\leq |\X|^\frac{1}{2}\Bigg[\mathbb{E}\Bigg(\sum_{k=N+1}^\infty F_k(\vec{\xi})^2\Bigg)^{1/2}+\sum_{k=1}^N \delta_k  \Bigg].
\end{align*}
Note that 
$\mathbb{E}\Bigg(\sum_{k=N+1}^\infty F_k(\vec{\xi})^2\Bigg) = \sum_{k=N+1}^\infty \mathbb{E}\big(F_k(\vec{\xi})^2\big) = \sum_{k=N+1}^\infty \lambda_k(C_f)$.
Therefore, we have our desired result
\[
\begin{aligned}
\mathbb{E}\{\bar{E}(\vec{\xi})\} &\leq 
|\X|^\frac{1}{2}\Bigg[\Bigg(\sum_{k=N+1}^\infty \lambda_k(C_f)\Bigg)^{1/2} + \sum_{k=1}^N \delta_k\Bigg]\\
&=
|\X|^\frac{1}{2}\Bigg[\Bigg(\sum_{k=N+1}^\infty \lambda_k(C_f)\Bigg)^{1/2} + \sum_{k=1}^N 
\left(\sum_{j = r_k+1}^\Np \lambda_j(\mat{S}_k)\right)^{1/2}\Bigg].
\end{aligned}. 
\]
\end{myproof}

\end{document}